\documentclass[letterpaper, oneside, 12pt]{article}
\usepackage{amsmath}
\usepackage{graphicx}%
\usepackage{amsfonts}%
\usepackage{amssymb}
\usepackage{xcolor}
\usepackage{overpic}
\usepackage[noend]{algpseudocode}
\usepackage{subfig}
\usepackage{rotating}
\usepackage{epstopdf}
\usepackage{url}
\usepackage{comment}
\usepackage{enumerate}
\usepackage{mathtools}
\usepackage{multirow}
\usepackage{pgfplots}
\usepackage{natbib}
\usepackage{pst-solides3d}
\usepackage[ruled,linesnumbered]{algorithm2e}
\pgfplotsset{compat=1.18} 

\usepackage{setspace}
\usepackage[hmargin=1.0in,vmargin=1.0in]{geometry}

\usepackage[compact]{titlesec}
\titlespacing{\section}{0pt}{*0.8}{*0.8}
\titlespacing{\subsection}{0pt}{*0.8}{*0.8}
\titlespacing{\subsubsection}{0pt}{*0.8}{*0.8}

\newcommand{\bD}{ {\boldsymbol D} }

\newcommand{\bI}{ {\boldsymbol I} }

\newcommand{\bJ}{ {\boldsymbol J} }

\newcommand{\bx}{ {\boldsymbol x} }
\newcommand{\bX}{ {\boldsymbol X} }

\newcommand{\bkappa}{ {\boldsymbol \kappa} }

\newcommand{\bzero}{ {\boldsymbol 0} }

\newcommand{\given}{\,|\,}

\newtheorem{theorem}{Theorem}[section]
\newtheorem{lemma}[theorem]{Lemma}

\newenvironment{proof}[1][Proof]{\begin{trivlist}
\item[\hskip \labelsep {\bfseries #1}]}{\end{trivlist}}
\newenvironment{definition}[1][Definition]{\begin{trivlist}
\item[\hskip \labelsep {\bfseries #1}]}{\end{trivlist}}

\newcommand{\qed}{\nobreak \ifvmode \relax \else
      \ifdim\lastskip<1.5em \hskip-\lastskip
      \hskip1.5em plus0em minus0.5em \fi \nobreak
      \vrule height0.75em width0.5em depth0.25em\fi}

\title{Differentially Private Estimation of Weighted Average Treatment Effects for Binary Outcomes}
\author{{\small Sharmistha Guha}\\
{\small Assistant Professor, Department of Statistics, Texas A\&M University,}\\ {\small 3143 TAMU, College Station, TX 77843-3143, E-mail: sharmistha@tamu.edu}\\
{\small Jerome P. Reiter}\\
{\small Professor, Department of Statistical Science, Duke University, }\\ {\small 214 Old Chemistry Building, Durham, NC 27708-0251, E-mail: jreiter@duke.edu}}

\begin{document}
\maketitle
\begin{abstract}
In the social and health sciences, researchers often make causal inferences using sensitive variables. 
These researchers, as well as the data holders themselves, may be ethically and perhaps legally obligated to protect the confidentiality of study participants' data. It is now known that releasing any statistics, including estimates of causal effects, computed with confidential data leaks information about the underlying data values. Thus, analysts may desire to use causal estimators that can provably bound this information leakage.  
Motivated by this goal, we develop algorithms for estimating weighted average treatment effects with binary outcomes that satisfy the criterion of differential privacy.
We present theoretical results on the accuracy of several differentially private estimators of weighted average treatment effects. We illustrate the empirical performance of these estimators using simulated data and a causal analysis using data on education and income.
\end{abstract}
\noindent\emph{Keywords:} Causal; Confidentiality; Observational; Privacy; Propensity.

\section{Introduction}

Many causal studies measure sensitive binary outcome variables that data stewards are ethically 
%, and sometimes even legally, 
obligated to keep confidential.  As hypothetical but realistic examples, the outcomes could be whether or not a patient  is cured of a stigmatized disease after treatment, a student passes or fails a test after receiving an intervention, or a person is employed or unemployed after job training.  In each of these cases, study participants would not want the data analyst to release information in a manner that reveals their individual outcomes.  Furthermore, causal studies typically include additional sensitive or identifying variables that analysts want to use as covariates; these too may be confidential.

Data stewards routinely put controls in place to reduce risks of unintended disclosures. For example, often they restrict access to the confidential data to vetted data analysts.  However, researchers in data privacy have shown that every  statistic computed with confidential data leaks information about that data \citep{dwork2014algorithmic, exposed}. 
%Typical methods include de-indentifying records even among the analysis team.
Given enough released information of sufficient accuracy, ill-intentioned users may be able to learn confidential information.  Thus, data stewards and analysts may seek to bound the information leakage when sharing results of confidential data analysis. 

One way to do so is to require 
%In response to these risks, some data stewards have turned to 
methods to provide formal guarantees of confidentiality protection for any data release.  Among such methods, algorithms that satisfy differential privacy \citep{dwork2006differential, dmns} have become a gold standard.  Differential privacy is a mathematical criterion that encodes the idea that the released statistic should not be overly sensitive to the presence or absence of any particular individual. Researchers have developed differentially private algorithms for a variety of estimation tasks,  including significance tests \citep[e.g., ][]{barrientos2019differentially, balle2020hypothesis, pensia2023simple}, regression \citep[e.g., ][]{zhang2012functional, wang2015regression, fang2019regression, gaboardi2019locally}, and machine learning \citep[e.g., ][]{mivule2012towards, ji2014differential, abadi2016deep, triastcyn2020bayesian, zheng2020preserving, blanco2022critical}, among many others.  

%Surprisingly, t
The literature includes few approaches to differentially private causal inference, particularly in observational studies.  \citet{dorazio} present differentially private algorithms for estimating the differences of two means in matched pairs designs, which are common in causal inference.  \citet{leeetal} construct a differentially private inverse-probability weighting treatment effect estimator.  They first fit a differentially private propensity score model to determine the weights, and then add Gaussian noise under $(\epsilon, \delta)$ differential privacy to perturb the resulting weighted treatment effect estimate.  Their method does not provide standard errors or interval estimates for the treatment effect. \citet{niu} use partitions of the data to estimate parts of machine learning algorithms that are used for causal inference, and combine all the parts together at the end to arrive at a causal estimate.  Their method also does not provide standard errors or interval estimates.  Finally, \citet{ohnishi:awan} show that plugging-in differentially private versions of parts of causal estimators can result in biased estimates. They also provide a Bayesian version of causal inference under the local differential privacy model, i.e., when individuals perturb their own data before providing it to a central party that does computations.

In this article, we contribute to this literature by proposing differentially private algorithms for causal inference with binary outcomes. The algorithms can be used with a variety of weighted average treatment effect estimators. Unlike other approaches, they generate standard errors and confidence intervals for these estimators.  The basic idea is to split the data into $M$ disjoint groups, estimate causal effects and standard errors in each partition, aggregate the results, and add differentially private noise to the aggregated results.  We illustrate the approach using simulation studies and an analysis of data from the 1994 U.S.\ census in a study of the effect of education on earnings.

The remainder of this article is organized as follows.  In Section \ref{sec: causal_differential}, we review key concepts from causal inference and differential privacy.
In Section \ref{sec:DPwate}, we present the differentially private treatment effect point and interval estimators.  In Section \ref{sec:sim_study}, we present results of simulation studies showing the performance of the proposed methodology in various scenarios. In Section \ref{sec:real_data}, we  illustrate the methodology using the 1994 U.S.\ census data.
Finally, in Section \ref{sec:conc}, we conclude with a discussion.

\section{Review of Causal Inference and Differential Privacy}\label{sec: causal_differential}
\noindent 
In Section \ref{sec: causal_inference}, 
we introduce the weighted average treatment effect (WATE) and methods for estimating the WATE. In Section \ref{sec: differential_privacy}, we review differential privacy and several algorithms that satisfy it.  Throughout the article, we suppose sample sizes to be large enough that  large sample approximations to sampling distributions are empirically valid.

\subsection{Overview of WATE Estimation}\label{sec: causal_inference}
We use the potential outcome framework for causal inference \citep{rubin1974estimating}. Let $z=1$ and $z=0$ indicate assignment to the treatment and control conditions, respectively. Let $y$ be an outcome variable.  We seek to learn the causal effect of $z$ on $y$. % subject in the study population is the comparison of two ``potential outcomes," 
For any unit in the study population, we conceive of two potential outcomes, $y(1)$ and $y(0)$, corresponding to the outcome measured when $z=1$ and $z=0$, respectively. For any unit, we observe only one of $y(1)$ and $y(0)$, which we write as $y=z y(1)+(1-z)y(0)$. We consider $y(0)$ and $y(1)$ as binary outcomes, i.e., $y(0),y(1)\in\{0,1\}$. We assume the stable unit treatment value assumption (SUTVA) which contains two sub-assumptions, no interference between units (i.e., the treatment applied to one unit does not affect the outcome for another unit) and no different versions of a treatment \citep{rubin1974estimating}. We also define the $p\times 1$ vector of covariates $\bx$, which are variables unaffected by treatment assignment $z$.  We assume that $0<P(z=1|\bx)<1$, i.e., the probability of assigning treatment or control is positive for every unit. Finally, we assume  
strong ignorability 
\citep{rosenbaum1983central} so that the vector of potential outcomes $(y(0), y(1))$ is independent of $z$ given $\bx$.

Many causal inference procedures utilize propensity scores $P(z=1|\bx)$, i.e., 
%which we write as $e(\bx)=P(z=1|\bx)$ for a given $\bx$.  
the probability of assignment to the treatment group given the covariates $\bx$. 
As shown by \cite{rosenbaum1983central}, the treatment assignment is independent of $\bx$ given $P(z=1|\bx)$ under SUTVA and strong ignorability. Propensity scores are typically  estimated using binary regressions of $z$ on $\bx$.  These estimated scores are used in a variety of causal estimators, especially in defining weighted sums of the outcomes as treatment effect estimators, 
%as use functions inverse probability weighting and overlap weighting, 
as we use here. In what follows, we  refer to estimated propensity scores using $e(\bx)$.

To compare outcomes under treatment and control, we define the conditional average controlled difference for a given $\bx$,  \begin{align}\label{causal_estimand}
\tau(\bx)=E[y|z=1,\bx]-E[y|z=0,\bx].
\end{align}
Under strong ignorability, $E[y(z)|\bx]=E[y|\bx,z]$, so that $\tau(\bx)$ in (\ref{causal_estimand}) becomes the average treatment effect conditional on $\bx$, i.e.,  $\tau(\bx)=E[y(1)-y(0)|\bx]$. Typically, 
%in either descriptive or causal comparisons, 
the (potential) outcomes are compared not for a single $\bx$; rather, they are  
averaged over a hypothesized target distribution of the covariates. The choice of the distribution corresponds to the region of the covariate space for the target population of interest. For example, if one seeks to estimate the effect of the treatment on the treated, the relevant covariate distribution is that of the treated cases.

Let the marginal density of $\bx$ be $f(\bx)$, defined with respect to a base measure $\Delta(\bx)$ (a product of counting measure for categorical variables and Lebesgue measure
for continuous variables).  %\citet{li2018balancing} show that,
%For many populations typically of interest in causal inference, 
For many common target populations in causal inference, the distribution of the covariates can be represented as $g(\bx)=f(\bx)t(\bx)$, where $t(\cdot)$ is a pre-specified function of $\bx$; some examples are provided below. Using this expression, we define a general class of estimands by the expectation of the conditional average controlled difference over the target population,
\begin{align}\label{tau1}
\tau =\frac{\int \tau(\bx)t(\bx)f(\bx)\Delta(d\bx)}{\int t(\bx)f(\bx)\Delta(d\bx)}.
\end{align}

The class of estimators defined in (\ref{tau1}) is referred to as the weighted average treatment effect (WATE) for causal comparisons \citep{hirano2003efficient}.
Specification of $t(\cdot)$ defines the target population and WATE estimands. Here, we consider three different WATEs. %corresponding to different choices of $t(\bx)$.  
When $t(\bx)=1$, the corresponding target population is the combined (treated and control) population, and the estimand is the average treatment effect (ATE). When $t(\bx) = e(\bx)$, the target population is the treated subpopulation,
and the estimand is the average treatment effect for the treated (ATT).  Finally, when $t(\bx) = 1-e(\bx)$, the target population is the control subpopulation,
and the estimand is the average treatment effect for the control (ATC). 

For any unit $i$ in a study where $i=1, \dots, n$, let their covariates be $\bx_i$, their treatment status $z_i$, and their outcome $y_i=z_i y_i(1)+(1-z_i)y_i(0)$.  The observed data are $\bD=\{(y_i,\bx_i,z_i):i=1, \dots, n\}$. We refer to the sampled covariate values as $\bX=\{\bx_1, \dots, \bx_n\}$. For $i=1, \dots, n$, let  $w_{1i} = t(\bx_i)/e(\bx_i)$, and let $w_{0i} = t(\bx_i)/(1-e(\bx_i))$. 
A consistent estimator of $\tau$ for any target population represented by the function $t(\cdot)$ is given by
\begin{align}\label{tau1est}
\hat{\tau} = \frac{\sum_{i=1}^n w_{1i} z_i y_i}{\sum_{i=1}^n w_{1i} z_i} - \frac{\sum_{i=1}^n w_{0i} (1-z_i) y_i}{\sum_{i=1}^n w_{0i} (1-z_i)}.
\end{align}
Expressions for $w_{0i}, w_{1i}$ and $\hat{\tau}$ corresponding to the ATE, ATT and ATC are given in Table~\ref{Tab_first}. We denote the treatment effect estimators as 
$\hat{\tau}_{ATE}, \hat{\tau}_{ATT}, \hat{\tau}_{ATC}$, respectively. 
\begin{table}[t]
\small
\begin{center}
\begin{tabular}{cccc} 
\hline
%& & & &\\
$\tau$ & ($w_{0i}$, $w_{1i}$) & Estimator & Estimated variance\\ 
\hline
& & &\\
ATE & \Big($\frac{1}{1-e(\bx_i)}$, $\frac{1}{e(\bx_i)}$\Big) & $\frac{\sum_{i=1}^n \frac{z_i y_i}{e(\bx_i)}}{\sum_{i=1}^n \frac{z_i}{e(\bx_i)}} - \frac{\sum_{i=1}^n  \frac{(1-z_i) y_i}{1-e(\bx_i)}}{\sum_{i=1}^n \frac{(1-z_i)}{1-e(\bx_i)}}$  & $\frac{\sum_{i=1}^n \left\{\frac{v_1(\bx_i)}{e(\bx_i)} + \frac{v_0(\bx_i)}{1-e(\bx_i)}\right\}}{n^2}$\\
& & &\\
%\hline
ATT &\Big($\frac{e(\bx_i)}{1-e(\bx_i)}$, 1\Big)  & $\frac{\sum_{i=1}^n z_i y_i}{\sum_{i=1}^n z_i} - \frac{\sum_{i=1}^n  \frac{(1-z_i) y_i e(\bx_i)}{1-e(\bx_i)}}{\sum_{i=1}^n \frac{(1-z_i)e(\bx_i)}{1-e(\bx_i)}}$  & $\frac{\sum_{i=1}^n e(\bx_i)^2 \left\{\frac{v_1(\bx_i)}{e(\bx_i)} + \frac{v_0(\bx_i)}{1-e(\bx_i)}\right\}}{\left[\sum_{i=1}^n e(\bx_i))\right]^2}$\\ 
& & & \\
%\hline
%\cline{2-5} 
ATC &\Big(1, $\frac{1-e(\bx_i)}{e(\bx_i)}$\Big) & $\frac{\sum_{i=1}^n \frac{z_i y_i(1-e(\bx_i))}{e(\bx_i)}}{\sum_{i=1}^n \frac{z_i (1-e(\bx_i))}{e(\bx_i)}} - \frac{\sum_{i=1}^n (1-z_i) y_i}{\sum_{i=1}^n (1-z_i)}$  & $\frac{\sum_{i=1}^n (1-e(\bx_i))^2 \left\{\frac{v_1(\bx_i)}{e(\bx_i)} + \frac{v_0(\bx_i)}{1-e(\bx_i)}\right\}}{\left[\sum_{i=1}^n (1-e(\bx_i))\right]^2}$\\ 
& & & \\
\hline
\end{tabular}
\caption{The expressions for $w_{0i}$, $w_{1i}$, the estimated treatment effect and its estimated variance for different choices of target population represented by $t(\bx_i)$.  ATE is the average treatment effect for everyone ($t(\bx)=1$). ATT is  the average treatment effect for the treated  ($t(\bx) = e(\bx)$). ATC is the average treatment effect for the controls ($t(\bx) = 1-e(\bx)$).  %ATO is the average treatment effect for the overlapping population ($t(\bx)=e(\bx)(1-e(\bx))$).
}\label{Tab_first}
\end{center}
\end{table}

For any of these estimators, which we write generically as $\hat{\tau}$, we can approximate the variance $V[\hat{\tau}]$ using the decomposition, $V[\hat{\tau}] = E_{\bx} V[\hat{\tau}| \bX] + V_{\bx} E[\hat{\tau}| \bX]$.
\cite{li2018balancing} derive the component of variation due to residual (model) variation conditional on $\bX$. Specifically, if 
$\mbox{V}[y(1)|\bX]=v_1(\bx)$ and $\mbox{V}[y(0)|\bX]=v_0(\bx)$ denote the variances of the outcome given the covariates for  the treated and control groups, respectively, \cite{li2018balancing}  show that $E_{\bx} V[\hat{\tau}| \bX]$ 
%conditional on the covariate design points 
can be approximated by 
\begin{align}\label{variance}
V=\frac{1}{n}\int t(\bx)^2\left\{\frac{v_1(\bx)}{e(\bx)}+\frac{v_0(\bx)}{(1-e(\bx))}\right\}f(\bx)\Delta(\bx)/\left\{\int t(\bx)f(\bx)\Delta(\bx)\right\}^2,
\end{align}
when the sample size $n$ is large.  \cite{imbens2004nonparametric} shows that  
%the individual variance is typically much larger than conditional mean variance, i.e., 
$E_{\bx} V[\hat{\tau}| \bX]$ is typically much larger than $V_{\bx} E[\hat{\tau}| \bX]$. Therefore, the general strategy is to approximate $V[\hat{\tau}]$ by (\ref{variance}). The expression in (\ref{variance}) can be  empirically approximated by,
\begin{align}\label{vareq}
\hat{V}=\frac{\frac{1}{n}\left[\frac{1}{n}\sum_{i=1}^n t(\bx_i)^2 \left\{\frac{v_1(\bx_i)}{e(\bx_i)} + \frac{v_0(\bx_i)}{1-e(\bx_i)}\right\}\right]}{\left[\frac{1}{n}\sum_{i=1}^n t(\bx_i))\right]^2} 
= \frac{\left[\sum_{i=1}^n t(\bx_i)^2 \left[\frac{v_1(\bx_i)}{e(\bx_i)} + \frac{v_0(\bx_i)}{1-e(\bx_i)}\right]\right]}{\left[\sum_{i=1}^n t(\bx_i))\right]^2}.
\end{align}
The estimated variance $\hat{V}$ corresponding to $\hat{\tau}_{ATE}, \hat{\tau}_{ATT}$ and $\hat{\tau}_{ATC}$ are denoted $\hat{V}_{ATE}, \\\hat{V}_{ATT}$ and $\hat{V}_{ATC},$ respectively. Their expressions are in Table~\ref{Tab_first}. Note that $\hat{V}$ is a consistent estimator of $V$.

With large $n$, 95\% confidence intervals for $\tau$ are constructed based on a large-sample normal approximation, $(\hat{\tau}-1.96\sqrt{\hat{V}},\hat{\tau}+1.96\sqrt{\hat{V}})$. 

In some settings, values of $e(\bx_i)$ can be close to zero or one, which can result in inflated variances.  In such cases, one remedy is to replace $e(\bx)$ with a truncated propensity score, given by
\begin{align}\label{eq:truncated_propen}
e^{T}(\bx)=\left\{\begin{array}{cc}
1-a & \mbox{if}\:\:e(\bx_i)>1-a\\
e(\bx) & \mbox{if}\:\:a\leq e(\bx_i)\leq 1-a\\
a & \mbox{if}\:\:e(\bx_i)<a\\
\end{array}
\right.
\end{align}
for some user-defined parameter $0< a < 1/2$.  The value of $a$ typically is chosen to be small, so that truncation affects only the few units with $e(\bx_i)$ near zero or one. 
Inferences are based on the expressions in Table \ref{Tab_first} replacing $e(\bx_i)$ with $e^T(\bx_i)$.  We denote the causal estimators based on truncated propensity scores %corresponding to the combined, treated, and control populations 
as  $\hat{\tau}_{ATE}^T, \hat{\tau}_{ATT}^T,$ and $\hat{\tau}_{ATC}^T$, with 
%respectively, upon truncation. The
corresponding variance estimates  $\hat{V}_{ATE}^T$, $\hat{V}_{ATT}^T$, and $\hat{V}_{ATC}^T$. 

Trimming estimated propensity scores is a common strategy in the causal inference literature,  
%(without any context of differential privacy), 
especially to avoid large variance and poor finite-sample performance due to large values of $w_{1i}$ or $w_{0i}$ \citep{kang2007demystifying}. The idea was first discussed in medical applications \citep{vincent2002anemia,grzybowski2003mortality,kurth2006results} and  formalized by \cite{crump2009dealing}.
%who suggested dropping units from the analysis with estimated propensity score outside an interval, so that the average treatment effect for the target population can be estimated with the smallest asymptotic variance.  

\subsection{Differential Privacy: Overview of Key Concepts}\label{sec: differential_privacy}
\begin{definition}\label{DPdef}($\epsilon$-Differential Privacy) An algorithm $\mathcal{P}$ satisfies $\epsilon-$differential privacy (denoted as $\epsilon$-DP), if for any pair of neighboring databases $(\bD,\bD')$, and any non-negligible measurable set $S\subseteq \mbox{range}(\mathcal{P})$, 
$P(\mathcal{P}(\bD)\in S) \leq \exp(\epsilon) P(\mathcal{P}(\bD')\in S)$. 
\end{definition}
Thus, $\mathcal{P}$ satisfies $\epsilon$-DP when the distributions of its outputs are similar for any two
neighboring databases, where the factor $\exp(\epsilon)$ defines the similarity. The $\epsilon$, known as the privacy loss  budget, controls the degree of confidentiality protection provided by $\mathcal{P}$, with greater protection guarantees implied by lower values. $\epsilon$-DP is a strong criterion, because an attacker who has access to
all of $\bD$ except any one row learns little from $\mathcal{P}(\bD)$ about the values in that unknown row when
$\epsilon$ is small \citep{ barrientos2019differentially}. 

The definition of $\epsilon$-DP satisfies several desirable properties.  Let $\mathcal{P}_1$ and $\mathcal{P}_2$ be $\epsilon_1$-DP and $\epsilon_2$-DP algorithms. The first is sequential composition: for any database $\bD$, release of both $\mathcal{P}_1(\bD)$ and $\mathcal{P}_2(\bD)$ satisfies $(\epsilon_1+\epsilon_2)$-DP. This means that we are able to calculate the total privacy leakage from releasing multiple statistics. The second is parallel composition.  Let $\bD_1$ and $\bD_2$ be two data files on disjoint sets of individuals. Release of both $\mathcal{P}_1(\bD_1)$ and $\mathcal{P}_2(\bD_2)$ satisfies max$\{\epsilon_1,\epsilon_2\}$-DP. The third is the post-processing property. For any algorithm $\mathcal{P}_2$ not depending on $\bD$, releasing $\mathcal{P}_2(\mathcal{P}_1(\bD))$ for any $\bD$ satisfies $\epsilon_1$-DP. In other words, post-processing the output of an $\epsilon_1$-DP algorithm does not imply any additional privacy loss.

A commonly used method to ensure $\epsilon$-DP is the Laplace mechanism. Let $q(\bD)$ be a function that takes $\bD$ as an input and outputs some statistic in $\mathbb{R}^k$. For example, $q$ might sum the elements of one of the columns in $\bD$.  We define the global sensitivity $s(q,||\cdot||)=\max\limits_{\bD,\bD',d(\bD,\bD')=1}||q(\bD)-q(\bD')||$, where $d(\bD,\bD')=1$ implies that the two databases differ by only one row, and $||\cdot||$ is a norm specific to the context.  The Laplace mechanism computes 
 %an $\epsilon$-differentially private version of $f(\bD)$, 
$LM(\bD)=q(\bD)+\bkappa$, where $\bkappa$ is a $k\times 1$ vector of independent draws from a Laplace distribution with density
$g(x|\lambda)=(2\lambda)^{-1}\exp(-|x|/\lambda)$, where
$\lambda=s(q,||\cdot||)/\epsilon$ \citep{dwork2006differential}. In Section \ref{sec:DPwate}, we use the Laplace mechanism to construct causal estimators and associated 95\% confidence intervals satisfying $\epsilon$-DP.

As part of our developments in subsequent sections, we use the \emph{subsample and aggregate algorithm}  \citep{nissim2007smooth} to satisfy $\epsilon$-DP.  The algorithm consists of a sampling step and an aggregating step. In the sampling step, we partition the confidential data  $\bD$ into $M$ disjoint subsets ${\bD_1, \dots, \bD_M}$ and compute $q(\bD_m)$ in each $\bD_m$. In the aggregation step, we compute the average  $q(\bD_1, \dots, \bD_M) = \sum_{m=1}^M q(\bD_m)/M$. For many $q$, any single observation affects the output from at most one of the partitions, i.e., the one it is randomly assigned to.  For such $f$, 
%For typical $f$, 
the  global sensitivity of $q(\bD_1, \dots, \bD_M)$ generally is $1/M$ times that of $q(\bD)$. Using this sensitivity, we apply the Laplace mechanism to $q(\bD_1, \dots, \bD_M)$  to create the differentially private statistic.  The reduced  
sensitivity decreases the variance of the noise distribution, which in turn offers potential for increased accuracy of released results.

\section{Differentially Private Estimation of WATE}\label{sec:DPwate}

We now construct a 
differentially private estimator of $\tau$  and its associated 95\% interval estimate for the three target populations reviewed in Table \ref{Tab_first}.  Our general strategy is to (i) find expressions for global sensitivities of the point and variance estimators, (ii) use the subsample and aggregate algorithm with these global sensitivities to generate noisy versions of the point and variance estimates, and (iii) apply a Bayesian inferential procedure to turn these noisy quantities into an interval estimate for $\tau$.  

We begin by finding a global sensitivity of $\hat{\tau}$ for binary outcomes.  Determining a sharp bound on the sensitivity is tricky, since changing any one data point can change not only the outcomes but the propensity score estimation and hence weights for all individuals in \eqref{tau1est}.  Instead, we use the coarse bound shown in Lemma \ref{lem1}.

\begin{lemma}\label{lem1}
For $y_i(0),y_i(1)\in\{0,1\}$, if    
$0 < e(\bx_i) < 1$ for all $i=1, \dots, n$,  
%$w_{1i}>0$ for at least one $z_i=1$ and $w_{0i}>0$ for at least one $z_i=0$. T
the global sensitivity  of $\hat{\tau}$ in \eqref{tau1est} for the ATE, ATT and ATC is bounded by 2, i.e., 
$s(\hat{\tau},|\cdot|)\leq 2$. %\textcolor{red}{Do we have to assume at least one (or all) $w_{1i}>0$ when $z_i=1$ and at least one (all) $w_{0i}>0$ when $z_i=0$?}
\end{lemma}
\begin{proof}
For any $\bD$, we have   
\begin{align}
& |\hat{\tau}(\bD)|\leq  \left|\frac{\sum_{i = 1}^{n}w_{1i} z_i y_i(1)}{\sum_{i = 1}^{n} w_{1i} z_i}-\frac{\sum_{i = 1}^{n}w_{0i} (1-z_i) y_i(0)}{\sum_{i = 1}^{n} w_{0i} (1-z_i)}\right|\nonumber\\
& \leq\max\{\max_{i=1:n}y_i(1)- \min_{i=1:n}y_i(0),\max_{i=1:n}y_i(0)- \min_{i=1:n}y_i(1)\}\leq 1,
\end{align}
since $y_i(0),y_i(1)\in\{0,1\}$, $w_{0i}>0$ for at least one $z_i=0$, and $w_{1i}> 0$ for at least one $z_i=1$. Thus, for any two neighboring datasets $\bD$ and $\bD'$, we have  
$s(\hat{\tau},|\cdot|)\leq |\hat{\tau}(\bD)|+|\hat{\tau}(\bD')| \leq 2$.
\end{proof}

To ensure all $0<w_{0i},w_{1i}<\infty$ for $i=1, \dots, n$, we use truncated propensity scores for the differentially private WATE.  The truncation limit $a$ is set before looking at values in $\bD$, so that $a$ is not data-dependent.
Under truncation, 
the global sensitivities for the estimated treatment effects remain bounded by two. The truncation also facilitates  bounding the sensitivities for the estimated variances, as we now show.
\begin{theorem}\label{th:variance}
For $y_i(0),y_i(1)\in\{0,1\}$, bounds on the global sensitivities for $\hat{V}_{ATE}^T$, $\hat{V}_{ATT}^T$ and $\hat{V}_{ATC}^T$ %and $\hat{V}_{ATO}^T$ 
are as follows: 
 $s(\hat{V}_{ATE}^T,|\cdot|)\leq \frac{1}{an}$;
 $\,\,\,s(\hat{V}_{ATT}^T,|\cdot|)\leq \frac{1}{2a^2n}$; and
 $\,\,\,s(\hat{V}_{ATC}^T,|\cdot|)\leq \frac{1}{2a^2n}$. %and, 
 %$s(\hat{V}_{ATO}^T,|\cdot|)\leq \frac{1}{2a^2n}$.
\end{theorem}
\begin{proof}
For any $\bD$, we have 
\begin{align}
& \hat{V}_{ATE}^T(\bD) = \frac{\sum_{i=1}^n\left[\frac{v_1(\bx_i)}{e^T(\bx_i)}+\frac{v_0(\bx_i)}{1-e^T(\bx_i)}\right]}{n^2}
\leq \frac{n (\frac{2}{4a})}{n^2} = \frac{1}{2an},
\end{align}
where the  inequality follows because $v_0(\bx_i),v_1(\bx_i)\leq 1/4$ and $\min\{e^T(\bx_i),1-e^T(\bx_i)\}\geq a$. Hence, for any two neighboring  $\bD$ and $\bD'$, we have  
$s(\hat{V}_{ATE}^T,|\cdot|)\leq |\hat{V}_{ATE}^T(\bD)|+|\hat{V}_{ATE}^T(\bD')|\leq\frac{1}{an}$. Applying similar logic for $\hat{V}_{ATT}^T(\bD)$, we have  
\begin{align}
&\hat{V}_{ATT}^T(\bD) =\frac{\sum_{i=1}^n e^T(\bx_i)^2\left[\frac{v_1(\bx_i)}{e^T(\bx_i)}+\frac{v_0(\bx_i)}{1-e^T(\bx_i)}\right]}{\left\{\sum_{i=1}^n e^T(\bx_i)\right\}^2}
 \leq \frac{\sum e^T(\bx_i) \left[1/4 + \frac{(1-a)}{4a}\right]}{\left\{\sum_{i=1}^n  e^T(\bx_i)\right\}^2}
= \frac{1/4 + \left[\frac{(1-a)}{4a}\right] }{\left[\sum_{i=1}^n  e^T(\bx_i)\right]}
\leq \frac{1}{4na^2}.
\end{align}
Hence,
$s(\hat{V}_{ATT}^T,|\cdot|)\leq |\hat{V}_{ATT}^T(\bD)|+|\hat{V}_{ATT}^T(\bD')|\leq\frac{1}{2na^2}$. Likewise, for $\hat{V}_{ATC}^T(\bD)$, we have 
\begin{align}
\hat{V}_{ATC}^T(\bD)&=\frac{\sum_{i=1}^n (1-e^T(\bx_i))^2 \left[\frac{v_1(\bx_i)}{e^T(\bx_i)} + \frac{v_0(\bx_i)}{1-e^T(\bx_i)}\right]}{\left[\sum_{i=1}^n  (1-e^T(\bx_i))\right]^2}
= \frac{\sum_{i=1}^n (1-e^T(\bx_i)) \left[\frac{v_1(\bx_i)(1-e^T(\bx_i))}{e^T(\bx_i)} + v_0(\bx_i)\right]}{\left[\sum_{i=1}^n  (1-e^T(\bx_i))\right]^2} \notag\\
%= \frac{\sum_{i=1}^n e(x_i) \left[v_1(x_i) + \frac{v_0(x_i) e(x_i)}{1-e(x_i)}\right]}{\left[\sum_{i=1}^n  e(x_i)\right]^2}
&\leq \frac{\sum_{i=1}^n (1-e^T(\bx_i)) \left[\frac{(1-a)}{4a} + 1/4\right]}{\left[\sum_{i=1}^n  (1-e^T(\bx_i))\right]^2}
= \frac{\left[\frac{(1-a)}{4a} + 1/4\right] }{\left[\sum_{i=1}^n  (1-e^T(\bx_i))\right]}\leq \frac{1}{4na^2}.
\end{align}
Hence,
$s(\hat{V}_{ATC}^T,|\cdot|)\leq |\hat{V}_{ATC}^T(\bD)|+|\hat{V}_{ATC}^T(\bD')|\leq\frac{1}{2na^2}$. 
\end{proof}

To avoid introducing a substantial bias in $\hat{\tau}$, we should make $a$ small, e.g., $a \leq 0.10$.   However, with small $a$, the global sensitivities in Theorem \ref{th:variance} could be large enough that, with small $\epsilon$, the noise variance in the Laplace mechanism is large compared to $\hat{V}$ itself, which could lead to undesirably wide confidence intervals.
%around $\hat{\tau}.

Therefore, rather than use Laplace mechanisms  for $\hat{\tau}$ and $\hat{V}$, we use the subsampling and aggregation technique reviewed in Section~\ref{sec: differential_privacy}. We split $\bD$ into $M$ disjoint subsets, $\{\bD_1, \dots, \bD_M\}$, of approximately equal size. In each $\bD_m$, we estimate propensity scores using only the data in that subset, and truncate them as in \eqref{eq:truncated_propen}.  Using the truncated propensity scores, in each $\bD_m$ we compute the treatment effect estimate $\hat{\tau}_{m}^T$ of interest and its approximated variance $\hat{V}_{m}^T$ using the expressions in Table \ref{Tab_first}, replacing each $e(\bx_i)$ with its truncated version $e^T(\bx_i)$. Finally, we average these estimates over the $M$ partitions to obtain, for the particular treatment effect estimate of interest,  
\begin{align}\label{eq:average_trt}
\bar{\tau}^T=\sum_{m=1}^M\hat{\tau}_{m}^T/M,\:\:\bar{V}^T=\sum_{m=1}^M\hat{V}_{m}^T/M. 
 \end{align}

The global sensitivities of 
$\bar{\tau}^T$ and $\bar{V}^T$ are $2/M$ and
$s(\hat{V}^T,|\cdot|)/M$, respectively. 

The estimators in \eqref{eq:average_trt} require at least two records with $z_i=1$ and $z_i=0$ in each partition.  When $\bD$ has large sample size and adequate numbers of treated and control records, this is practically a near certainty for reasonable values of $M$.  However, it is theoretically possible to sample a partition with too few treateds or controls.  We discuss this scenario in Section \ref{sec:conc}.

To complete the subsampling and aggregation algorithm, we add independent noise to each of the quantities in \eqref{eq:average_trt} using Laplace mechanisms. Suppose the total privacy budget is $\epsilon$. For $0<\pi<1$, we use $(1-\pi)\epsilon$ privacy budget for the treatment effect estimate and  $\pi\epsilon$ for the variance estimate. 
Specifically, we compute $\bar{\tau}^{T,\epsilon}=\bar{\tau}^T+\eta_1$,
where $\eta_1\sim Laplace(0,2/(M\epsilon(1-\pi)))$ and  
$\bar{V}^{T,\epsilon}=\bar{V}^T+\eta_2$,
where $\eta_2\sim Laplace(0,s(\bar{V}^T,|\cdot|)/(M\epsilon\pi))$. 

While $\bar{\tau}^{T,\epsilon}$ and $\bar{V}^{T,\epsilon}$ are differentially private, they may not be readily usable to make interpretable inferences about $\tau$.
In particular, $\bar{\tau}^{T,\epsilon}$ is not guaranteed to lie in $(-1,1)$, and $\bar{V}^{T,\epsilon}$ could be negative. Therefore, we use a Bayesian post-processing algorithm to turn the differentially private point and variance estimates into interpretable inferences about $\tau$. The basic idea is as follows.  Since the data analyst only has  $(\bar{\tau}^{T,\epsilon}, \bar{V}^{T,\epsilon})$, the analyst treats $(\bar{\tau}^{T}, \bar{V}^{T})$ as unknown quantities. The analyst draws many, say $L$, plausible values of the unobserved $(\bar{\tau}^T, \bar{V}^T)$ from their posterior distribution, and from each plausible value samples a value of $\tau$.  The $L$ draws of $\tau$  can be summarized for inferences.

We now offer the details of this post-processing step.  For clarity we introduce new notation $(\bar{\tau}, \bar{V})$ to represent the analyst's random variables for the unknown values of $(\bar{\tau}^T, \bar{V}^T)$. We specify two models $\mathcal{M}_1$ and $\mathcal{M}_2$ independently, given by 
\begin{align}\label{model}
&(\mbox{Model}\:\mathcal{M}_1):\:\:\:\:\bar{\tau}^{T,\epsilon}=\bar{\tau}+\zeta_2,\:\:
\zeta_2\sim \mbox{Laplace}(0,2/(M\epsilon(1-\pi))),\nonumber\\
&(\mbox{Model}\:\mathcal{M}_2):\:\:\:\:\bar{V}^{T,\epsilon}=\bar{V}+\zeta_1,\:\:
\zeta_1\sim \mbox{Laplace}(0,s(\bar{V}^T,|\cdot|)/(M\epsilon\pi)).
\end{align}
 
We assign  $\bar{\tau}\sim U(-1,1)$ and $\bar{V}\sim U(0,s(\bar{V}^T,|\cdot|)/2)$ prior distributions, where $s(\bar{V}^T,|\cdot|)/2$ is the upper bound from Theorem \ref{th:variance} for the treatment effect of interest. 
We estimate the  posterior distributions of  $\bar{\tau}$ and $\bar{V}$
using elliptical slice sampling \citep{nishihara2014parallel}.  Importantly, we do not use the confidential values of $\bar{\tau}^{T}$ and $\bar{V}^{T}$ in the sampling algorithms; we only use the differentially private statistics. 

We obtain $L$ post burn-in samples of $\bar{\tau}$, denoted as $\bar{\tau}^{*(1)}, \dots, \bar{\tau}^{*(L)}$, and of $\bar{V}$, denoted as $\bar{V}^{*(1)}, \dots, \bar{V}^{*(L)}$. For $l=1, \dots, L$, we draw a sample $\tilde{\tau}^{(l)}\sim N(\bar{\tau}^{*(l)}, \bar{V}^{*(l)})$.  
   
These $L$ draws represent an approximate posterior distribution for $\tau$.  We use $\tilde{\tau}^{\epsilon}=\sum_{l=1}^L\tilde{\tau}^{(l)}/L$ as the privacy-protected point estimator for $\tau$. We construct the 2.5\% and 97.5\% empirical quantiles 
$\tilde{\tau}_{lower}^{\epsilon}$ and $\tilde{\tau}_{upper}^{\epsilon}$ from $\tilde{\tau}^{(1)}, \dots, \tilde{\tau}^{(L)}$, and use  $(\tilde{\tau}_{lower}^{\epsilon}, \tilde{\tau}_{upper}^{\epsilon})$ as the privacy-protected 95\% interval estimate for $\tau$.  

In drawing values of $\tau$, we rely on large-sample normality  for  
the sampling distribution of $\bar{\tau}^T$ as defined in  \eqref{eq:average_trt}, namely  $\bar{\tau}^T \sim N(\tau, \bar{V}^T)$. With a diffuse prior on $\tau$, we have $\tau \sim N(\bar{\tau}^T, \bar{V}^T)$, which we can sample from to summarize inferences about $\tau$.  This presumes the sampling variability in $\bar{V}^T$ is negligible compared to $\bar{V}^T$ itself, which is generally reasonable and typically assumed in large-sample inference  \citep{rubin1987multiple}. In our setting, the analyst does not know $\bar{\tau}^T$ and $\bar{V}^T$; rather, the analyst has plausible draws of each.
%$\tilde{\tau}^{*(1)}, \dots, \tilde{\tau}^{*(L)}$ of it.  
Thus, we follow the strategy described in \citet{zhou:reiter} for Bayesian inference with plausible draws of unknown values: for each plausible draw $(\tilde{\tau}^{*(l)}, \bar{V}^{*(l)})$ of $(\bar{\tau}^T, \bar{V}^T)$, we sample a value of $\tau$ from $N(\bar{\tau}^{*(l)}, \bar{V}^{*(l)})$, and concatenate the draws for inferences about $\tau$.

The entire process for estimating $\tau$ is summarized in Algorithm \ref{alg1}. Theorem~\ref{thm_dp} formally states and proves that Algorithm \ref{alg1} is differentially private.

\begin{theorem}\label{thm_dp}
Algorithm \ref{alg1} satisfies $\epsilon-$differential privacy.
\end{theorem}
\begin{proof}
Since $\bar{\tau}^T$ has a global sensitivity of $2/M$, defining $\bar{\tau}^{T,\epsilon}=\bar{\tau}^T+ Laplace(0,2/(M\epsilon(1-\pi)))$ 
%mplies, by Definition \ref{DPdef}, that this 
is a $(1-\pi)\epsilon$-DP algorithm. Using a similar argument, $\bar{V}^{T,\epsilon}=\bar{V}^T+ Laplace(0,s(\bar{V}^T,|\cdot|)/(M\epsilon\pi))$ is a $\pi\epsilon$-DP algorithm. The Bayesian inference steps rely entirely on $(\bar{\tau}^{T,\epsilon}, \bar{V}^{T,\epsilon})$. By the post-processing property of differential privacy, they do not affect the privacy guarantee. Hence, releasing $\tilde{\tau}^{\epsilon}$ and  $(\tilde{\tau}_{lower}^{\epsilon}, \tilde{\tau}_{upper}^{\epsilon})$ from Algorithm \ref{alg1} is $(1-\pi)\epsilon+\pi\epsilon=\epsilon$-DP.
\end{proof}

\begin{algorithm}\label{alg1}
	\caption{Differentially Private WATE and its 95\% Interval Estimate}
	\label{alg1}
	\SetAlgoLined
	\KwIn{(1) $\bD:$ Dataset $\{y_i,\bx_i,z_i: i=1, \dots, n\}$ ; (2) $M:$ Number of partitions; \\ (3) $a:$ Truncation level; (4) $\epsilon:$ Privacy loss budget; (5) $\pi:$ fraction of privacy loss budget allocated to variance estimation.}
	\KwOut{(1) DP WATE estimate  $\tilde{\tau}^{\epsilon}$; (2) DP 95\% interval $(\tilde{\tau}_{lower}^{\epsilon}, \tilde{\tau}_{upper}^{\epsilon})$ for WATE}
 
 \Begin
{
	\tcc{\small{Step 1: Partition the data as a part of subsample and aggregation step. }}
	%$\mathcal{G}^{(t)} = $ \textbf{PartitionUpdate}$ \left( \tilde{\bTheta}^{(1,t-1)},\ldots,\tilde{\bTheta}^{(S,t-1)} \right) $\\
	Choose a random partition $\{\bD_1, \dots, \bD_M\}$ of  $\bD$\\

	\tcc{\small{Step 2: Compute WATE estimate and its estimated variance based on truncated propensity scores in each subset.}}
	\For{$ m \in 1:M $}
	 {

	 	Compute WATE estimate $\hat{\tau}_{m}^T$ and its approximated variance $\hat{V}_{m}^T$ using the truncated propensity score defined in (\ref{eq:truncated_propen})
	 	 from $\bD_m$.
	 }
	 \tcc{\small{Step 3: Add noise following Laplace mechanism.}}
    Compute the average of treatment effects $\bar{\tau}^T$ and its estimated average variance $\bar{V}^T$ following (\ref{eq:average_trt}).\\
    Generate $\eta_1\sim $ Laplace($0,2/(M\epsilon(1-\pi))$)
     and $\eta_2\sim$ Laplace($0,s(\bar{V}^T,|\cdot|)/(\epsilon\pi)$).\\ 
     Compute noisy versions $\bar{\tau}^{T,\epsilon}=\bar{\tau}^T+\eta_1$ and $\bar{V}^{T,\epsilon}=\bar{V}^{T}+\eta_2.$\\
     \tcc{\small{Step 4: Apply Bayesian post-processing steps.}}
     Fit models $\mathcal{M}_1$ and $\mathcal{M}_2$ independently, which are given by
\begin{align}
&(\mbox{Model}\:\mathcal{M}_1):\:\:\:\:\bar{\tau}^{T,\epsilon}=\bar{\tau}+\zeta_2,\:\:
\zeta_2\sim \mbox{Laplace}(0,2/(M\epsilon(1-\pi))),\nonumber\\
&(\mbox{Model}\:\mathcal{M}_2):\:\:\:\:\bar{V}^{T,\epsilon}=\bar{V}+\zeta_1,\:\:
\zeta_1\sim \mbox{Laplace}(0,s(\bar{V}^T,|\cdot|)/(M\epsilon\pi)),
\end{align}
	\For{$l \in 1:L$}
	 {
		Draw elliptical slice samples \citep{nishihara2014parallel} for $\bar{\tau}$	and $\bar{V}$, denoted by $\bar{\tau}^{(l)}$ and $\bar{V}^{(l)}$, respectively.\\
		Draw $\tilde{\tau}^{(l)}\sim N(\bar{\tau}^{(l)}, \bar{V}^{(l)})$.
	 }
	 Compute $\tilde{\tau}^{\epsilon}=\sum_{l=1}^L\tilde{\tau}^{(l)}/L$ and $\tilde{\tau}_{lower}^{\epsilon}=2.5\%$ empirical quantile, $\tilde{\tau}_{upper}^{\epsilon}=97.5\%$ empirical quantile.\\
	 \Return{ $\tilde{\tau}^{\epsilon}$, $(\tilde{\tau}_{lower}^{\epsilon}, \tilde{\tau}_{upper}^{\epsilon})$.}
}
\end{algorithm}

\subsection{Theoretical Study of the DP WATE}\label{theory}
In this section, we discuss some asymptotic properties of 
%the differentially private WATE,  $
$\tilde{\tau}^{\epsilon}$.  Here, we presume any bias in the treatment effect estimators  introduced by truncating propensity scores is negligible. This is reasonable when $a$ is small and all $e(\bx_i)$ are bounded away from $a$. If the truncation does produce a non-negligible bias,
%of say $\delta$, 
the results still support the statement that the differentially private WATE has a similar distribution as the WATE based on the trimmed propensity scores.

As discussed in Section~\ref{sec:DPwate}, when the sample size $n$ is large, 
the analyst's draws from $N(\bar{\tau},\bar{V})$ are samples from the posterior distribution of $\tau$, from which the analyst obtains $\tilde{\tau}^{\epsilon}$. 
Let $\tilde{g}^{\epsilon}$ denote this distribution, i.e.,  $\tilde{g}^{\epsilon}=N(\bar{\tau},\bar{V})$. When $n$ is large, the distribution of the non-private estimator $\hat{\tau}$ is approximately $g=$ N($\tau,V$). To assess the discrepancy between $\tilde{\tau}^{\epsilon}$ and $\hat{\tau}$, we develop an upper bound for the discrepancy between $g$ and $\tilde{g}^{\epsilon}$, given by $P(KL(\tilde{g}^{\epsilon},g)>c)$, for any $c>0,$ as a function of $M$, $\epsilon$ and $V$, as $n\rightarrow\infty$.  Here, $KL$ stands for Kullback-Leibler divergence.

For $m=1, \dots, M$, let $n_m$ be the number of observations in $\bD_m$; let $\bD_{T,m}$, $\bD_{u,m}$ and $\bD_{l,m}$ denote the sets of observations with propensity scores between $a$ and $1-a$, greater than $1-a$, and less than $a$, respectively; and, let $n_{T,m}$, $n_{u,m}$ and $n_{l,m}$ be the number of samples in each of $\bD_{T,m}$, $\bD_{u,m}$ and $\bD_{l,m}$, respectively.  We make the following assumptions about the sample sizes. 

\textbf{(A1)} As $n\rightarrow\infty$, $n_{m}\rightarrow\infty$ for all $m=1, \dots, M$.

\textbf{(A2)} As $n_m\rightarrow\infty$, $n_{u,m}/n_m\rightarrow 0$ and $n_{l,m}/n_m\rightarrow 0$.

\noindent Since $n_m=n_{T,m}+n_{l,m}+n_{u,m}$, assumptions (A1) and (A2) imply that $n_{T,m}/n_m\rightarrow 1$ as $n\rightarrow\infty$. 

\begin{lemma}\label{lemma_DP}
%Let $\mathcal{R}$ denote the random mechanism used to partition $\bD$. 
Under (A1) and (A2), the following results hold. 
\begin{enumerate}[(i)]
    \item $P(|\bar{\tau}-\tau|>c)\leq 2\exp(-M\epsilon(1-\pi)c/6)$, as $n\rightarrow\infty$, for any $c>0$.
\item $P(|\bar{V}-V|>c)\leq 2\exp(-M\epsilon\pi c/6)$, as $n\rightarrow\infty$, for any $c>0$.
\end{enumerate}
\end{lemma}
A proof of Lemma \ref{lemma_DP} is in the appendix.  
We use Lemma \ref{lemma_DP} to derive a bound on 
$P(KL(\tilde{g}^{\epsilon},g)>c)$. 
\begin{theorem}\label{theorem_DP}
Under (A1) and (A2), we have
\begin{align*}
P(KL(\tilde{g}^{\epsilon},g)>c)
%d(\tilde{\tau},\hat{\tau})
\leq  2\exp\left(-\frac{M\epsilon(1-\pi)\sqrt{2Vc}}{6\sqrt{3}}\right)+4\exp(-M\epsilon\pi V c/9),   
\end{align*}
as $n\rightarrow\infty$.
\end{theorem}
Theorem \eqref{theorem_DP} shows that, as $M$ or $\epsilon$ get large, the distance between the two quantities goes to $0$.  

\section{Simulation Studies}\label{sec:sim_study}
In this section, we illustrate repeated sampling properties of the differentially private point, variance, and interval estimates of  
the ATE, ATT and ATC.  We begin with simulations that set $M=100$ and $a=0.05$ for data of size $n=10000$ and privacy loss budget $\epsilon=1$.  We then vary simulation design parameters one at a time to investigate the sensitivity of the findings.  Section~\ref{sensitivity_M} varies $M$;  Section \ref{sensitivity_a} varies $a$;  Section \ref{sensitivity_n} varies $n$; and, Section~\ref{sensitivity_ep} varies  $\epsilon$. 

\subsection{Baseline Studies with $M=100,  a=0.05, n=10000, \epsilon=1$}\label{sim:design}

To create $\bD$ in any simulation run, we generate $n= 10000$ observations each measured on $p=4$ covariates, $\bx_i = (x_{i1}, x_{i2}, x_{i3}, x_{i4})$. We simulate $\bx_i\sim N(\bzero,(1-\rho)\bI+\rho\bJ)$, $0\leq \rho\leq 1$, where $\bJ$ is a $p\times p$ matrix with each entry as $1$. The covariance implies correlation of $\rho$ between any two predictors, and we set $\rho=0.2$ in all simulations.
%For convenience, we sample each $x_{ik}$ from independent 
%$\bX$ is generated such that $\bX = [\boldsymbol{1}: \bX_p]$, where each element of $\bX_p$ is simulated 
%standard normal distributions. 
For $i=1, \dots, n$, we generate a treatment status $z_i$ from a Bernoulli draw with probability $P(z_i = 1|\bx_i)$, where   
\begin{align}\label{treatment_assignment}
\mbox{logit} [P(z_i = 1|\bx_i)] = 0.1 + 0.2 \eta x_{i1} + 0.5 \eta x_{i2} - 0.25 \eta x_{i3} - 0.45 \eta x_{i4}. 
\end{align}
We vary $\eta \in \{2,4\}$ to change the level of overlap in the treatment and control samples, as shown in Figure~\ref{Overlap-change}. As $\eta$ increases, we observe increasingly sparse overlap in the propensity score distribution.

\begin{figure}[t]
  \centering
  \subfloat[$\eta=2$]{%
    \includegraphics[width=0.45\textwidth]{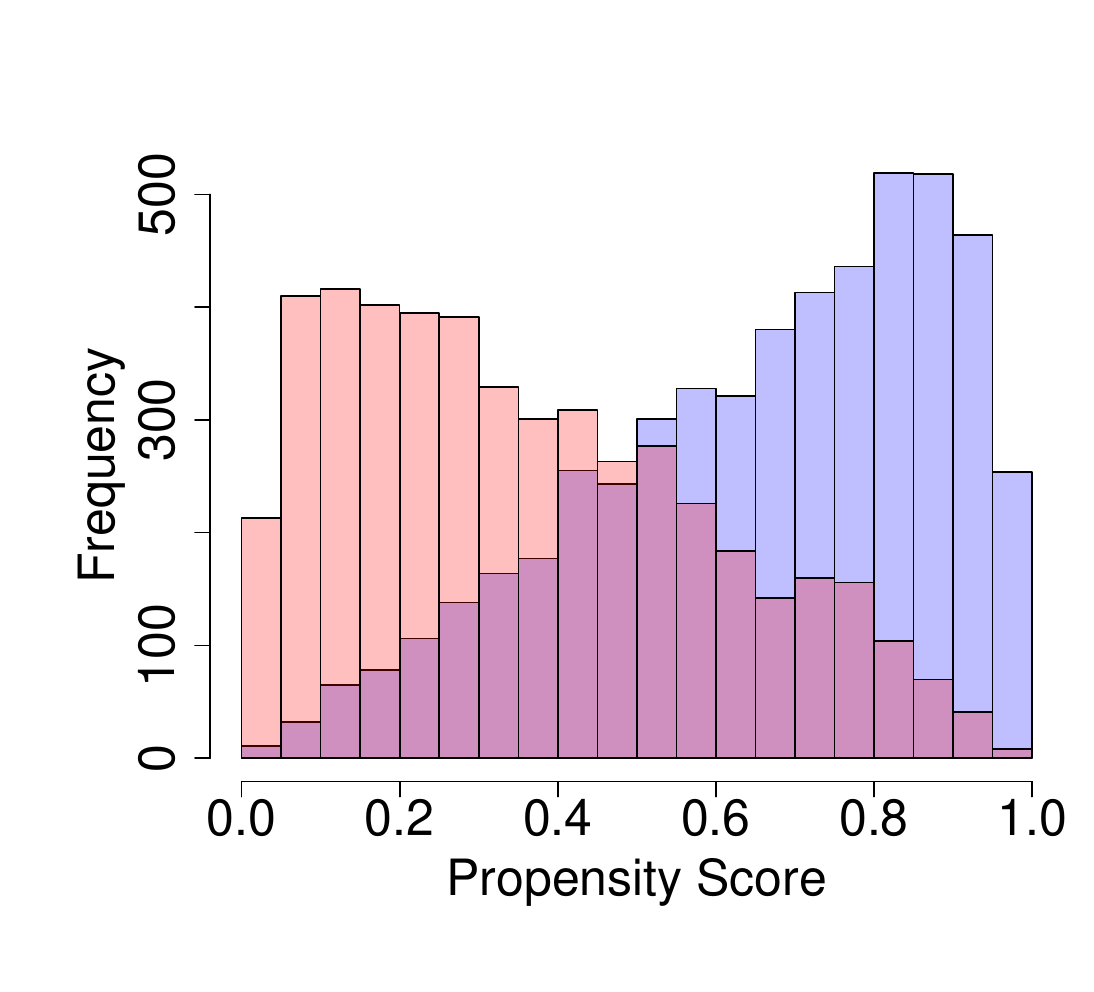}
    \label{overlap-nu2}
  }
  \hfill
  \subfloat[$\eta=4$]{%
    \includegraphics[width=0.45\textwidth]{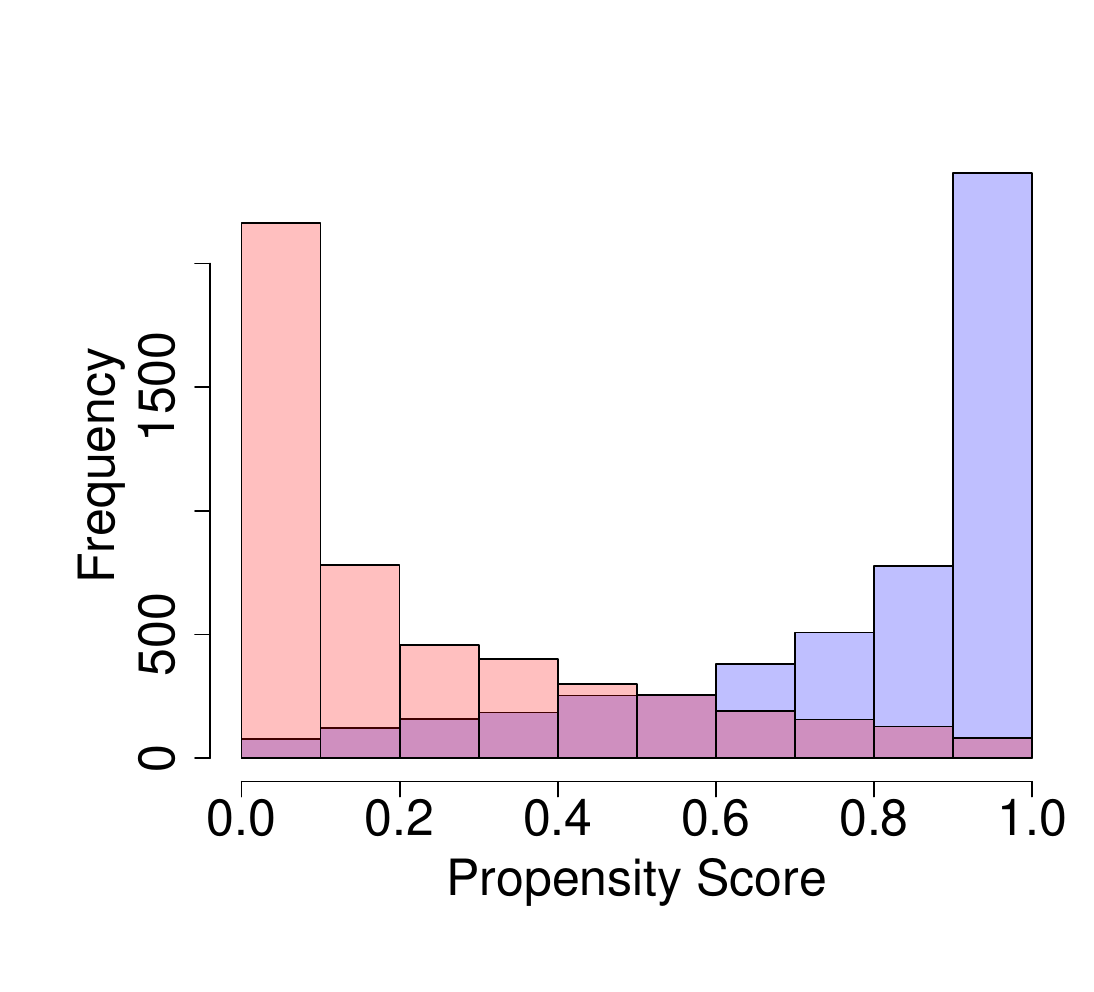}
    \label{overlap-nu4}
  }
  \caption{Simulated propensity score distributions in the simulation of Section \ref{sim:design} for treated (purple) and controls (pink). The propensity score distributions are shown for $\eta=2$ and $\eta=4$. The propensity score distributions show much less overlap when $\eta=4$.}
  \label{Overlap-change}
\end{figure}

For each $(\bx_i,z_i)$, we simulate the potential binary outcomes, $(y_i(0), y_i(1)),$ from Bernoulli distributions with probabilities governed by
\begin{align}\label{eq:outcome_gen}
\mbox{logit} [P(y(z) = 1)] = \beta_0 + \beta_1 x_1 + \beta_2 x_2 + \beta_3 x_3 + \beta_4 x_4 + \gamma z. 
\end{align}
We set $(\beta_0, \beta_1, \beta_2, \beta_3, \beta_4)=(0.15, -0.2, 0.3, -0.4, 0.6)$. 
We also vary $\gamma \in \{0,1,2\}$ to represent different strengths of treatment effects. Thus, we have six simulation scenarios corresponding to each combination of 
%$(\eta,\gamma)$, in particular 
$(\eta,\gamma)\in\{2,4\}\times\{0,1,2\}$.

In each simulation, we compute the true treatment effects,

\begin{align}
&\tau_{ATE} = (1/n) \sum_{i=1}^n \left\{P(y_i(1)=1 \given\bx_i) - P(y_i(0)=1 \given \bx_i)\right\} \label{simATE}\\
&\tau_{ATT} = (1/n_T) \sum_{i:z_i=1}\left\{ P(y_i(1)=1 \given \bx_i) - P(y_i(0)=1 \given\bx_i)\right\} \label{simATT}\\
&\tau_{ATC} = (1/n_C) \sum_{i:z_i=0} \left\{P(y_i(1)=1 \given \bx_i) - P(y_i(0)=1 \given \bx_i)\right\}, \label{simATC}
\end{align}
where the expressions for $P(y_i(1)=1\given \bx_i)$ and $P(y_i(0)=1\given \bx_i)$ are obtained from (\ref{eq:outcome_gen}). The quantities $n_T$ and $n_C$ are the number of treated and control subjects, respectively.

For comparisons, in each simulation run, we compute estimated treatment effects $\hat{\tau}$ for each causal estimand without privacy concerns, i.e., without the partitions, truncation, or Laplace noise. 
We also compute the 95\% confidence interval (CI) using $(\hat{\tau} \pm 1.96 \sqrt{\hat{V}})$, where $\hat{V}$ is the estimated variance of the treatment effect calculated on the sample using (\ref{vareq}) without any privacy protections. 
We run $500$ independent replications of each scenario, sampling a new set of $\{(\bx_1,z_1,y_1), \dots, (\bx_n,z_n,y_n)\}$ each time.  We equally allocate privacy budget in the WATE point and variance estimation, and set $\pi=0.5$.

\begin{table}[h]
%\small
\begin{center}
\begin{tabular}{clcccccc} 
\hline
%& \textbf{Scenario} & \textbf{1} & \textbf{2} & \textbf{3} & \textbf{4} & \textbf{5} & \textbf{6} \\
%\hline
\hline
& $\eta$ & 2 & 2 & 2 & 4 & 4 & 4 \\ 
\cline{2-8} 
& $\gamma$ & 0 & 1 & 2 & 0 & 1 & 2 \\
\cline{2-8} 
%& $\tau_0$ & 0 & .048  & .090 & 0 & .047 & .091 \\
\hline
\hline
& Avg. $\tau_{ATE}$  & 0 & .204 & .342 & 0 & .204 & .343\\
\multirow{6}{4em}{ATE} & $\hat{\tau}_{ATE}$ & & & & &\\
& $\,\,\,\,$ RMSE & .011 & .014 &  .012 & .019 & .019 & .022\\ 
& $\,\,\,\,$ 95\% CI coverage & 90.2 & 89.8 & 90.8 & 91.2 & 92.4 & 92.0 \\ 
& $\,\,\,\,$ 95\% CI length & .054 & .058 & .057 & .171 & .192 & .187 \\ 
%\cline{2-8} 
& $\tilde{\tau}^{\epsilon}_{ATE}$ \\
& $\,\,\,\,$ RMSE & .016 & .016 & .015 & .023 & .021 & .024 \\ 
& $\,\,\,\,$ 95\% CI coverage & 96.8 & 97.4 & 97.2 & 98.0 & 98.0 & 98.2 \\ 
& $\,\,\,\,$ 95\% CI length & .113 & .134 & .148 & .281 & .295 & .316 \\
%\hline
\hline
& Avg. $\tau_{ATT}$  & 0 & .205 & .345 & 0 & .206 & .348\\
\multirow{6}{4em}{ATT} &$\hat{\tau}_{ATT}$ \\
& $\,\,\,\,$ RMSE & .012 & .012 & .010 & .022 & .017 & .019\\ 
& $\,\,\,\,$ 95\% CI coverage & 89.8 & 90.8 & 91.6 & 91.0 & 91.8 & 92.0 \\ 
& $\,\,\,\,$ 95\% CI length & .059 & .062 & .062 & .201 & .229 & .213 \\ 
%\cline{2-8} 
& $\tilde{\tau}^{\epsilon}_{ATT}$ \\
&  $\,\,\,\,$ RMSE & .016 & .014 & .013 & .026 & .023 & .021 \\ 
&  $\,\,\,\,$ 95\% CI coverage & 96.6 & 96.8 & 97.4 & 97.0 & 98.2 & 98.0 \\ 
&  $\,\,\,\,$ 95\% CI length & .136 & .144 & .169 & .303 & .324 & .332 \\ 
%\hline
\hline
& Avg. $\tau_{ATC}$  & 0 & .202 & .338 &  0 & .202 & .337\\
\multirow{6}{4em}{ATC} & $\hat{\tau}_{ATC}$\\
& $\,\,\,\,$RMSE & .015 & .014 & .010 & .022 & .020 & .017 \\ 
&  $\,\,\,\,$ 95\% CI coverage & 90.2 & 90.4 & 91.0 & 91.6 & 92.2 & 92.4 \\ 
&  $\,\,\,\,$ 95\% CI length & .057 & .065 & .064 & .215 & .234 & .219 \\
%\cline{2-8} 
& $\tilde{\tau}^{\epsilon}_{ATC}$ \\
&  $\,\,\,\,$ RMSE & .018 & .018 & .015 & .028 & .025 & .026 \\ 
&  $\,\,\,\,$ 95\% CI coverage & 97.2 & 97.4 & 97.4 & 98.0 & 98.2 & 98.4 \\ 
&  $\,\,\,\,$ 95\% CI length & .132 & .159 & .179 & .310 & .326 & .341 \\ 
 
\hline
\end{tabular}
\caption{Results from 500 simulations with $(M=100, \epsilon=1, a=.05)$ for the ATE, ATT and ATC. Results include the root of the average squared errors (RMSE) between the differentially private WATE estimate and the corresponding true value based on  \eqref{simATE} -- \eqref{simATC}; the percentage of the five hundred 95\% confidence intervals that cover the corresponding treatment effect; and, the average length of the estimated 95\% confidence intervals in parenthesis.  These quantities are shown for both the privacy protected and non-privacy protected estimation.}\label{Tab1}
\end{center}
\end{table}

%\noindent\underline{\textbf{Simulation Results}}\\
As evident in Table~\ref{Tab1}, the differentially private point estimates have small average errors, indicating that they accurately estimate true treatment effects.  For context, the values of the various $\tau$ for scenarios with $\gamma=1,2$ tend to be around $.20$ to $.30$, so that RMSEs of $.01$ to $.02$ are modest fractions of the true treatment effects.  The average errors from the differentially private estimates tend to be larger than those from the non-privacy protected point estimates, reflecting the combined effects of the Laplace noise, truncation limits, and subsampling.  All methods are more accurate when there is greater overlap in the propensity scores, as one would expect. 

The  95\% CIs for the various $\tau$ without privacy protection cover less often than the nominal 95\% rate, whereas the 95\% intervals based on the  differentially private algorithms tend to cover more often than the nominal 95\% rate.   
The steps taken to protect privacy result in increased average interval lengths, which seemingly is the price to pay for the privacy protection.

\subsection{Sensitivity to the Choice of $M$}\label{sensitivity_M}

The choice of $M=100$ ensures that the standard deviation of the Laplace noise is much smaller than $\bar{V}^T$.  
In this section, we consider the effects on inference of changing $M$ by considering $M \in \{50, 100, 200\}$. In generating the data, we use the more challenging case of less overlap between treatment and control propensity score distributions, setting $\eta = 4$ and $\gamma = 1$. All other parameters are set at the  values described in Section \ref{sim:design}. 

Figure~\ref{M-change} displays the RMSEs of the point estimators, and the   
coverage rates and average lengths of the 95\% intervals. In these simulations, we see little practical impact of changing $M$  on the properties of the differentially private WATE estimates. 
For all $M$ considered, the differentially private WATE estimates generally are close to the corresponding non-private treatment effect WATE estimates computed with $\bD$, and  
the coverage rates of the differentially private intervals are all around 98\%. 
The interval lengths also are unremarkably different, with some suggestion that the interval lengths are smallest at $M=50$. As $M$ increases, the intervals are subject to  two countervailing effects. The variance of the Laplace noise decreases thereby encouraging shorter intervals, and the uncertainty in the propensity score estimates in each partition increases thereby encouraging longer intervals. We confirmed the latter fact by constructing 95\% intervals from samples of $\tau \sim N(\bar{\tau}^T,\bar{V}^T)$, that is, 
using partitioning without adding  Laplace noise. 
We suggest a rule-of-thumb for selecting $M$ in Section \ref{sec:conc}.

\begin{figure}[t]
  \centering
  \subfloat[]{\includegraphics[height=7cm,width=8cm]{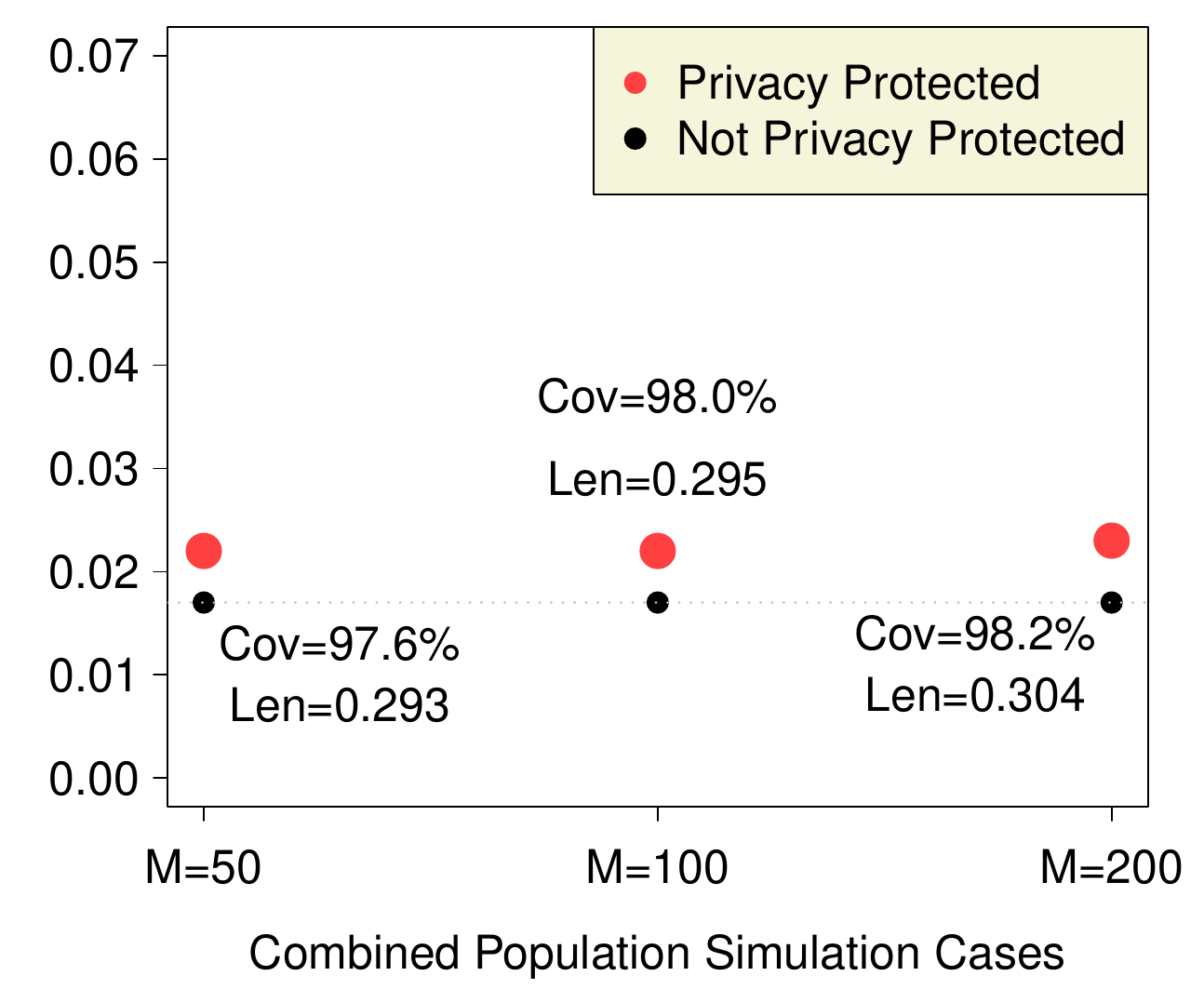}\label{sim-M1}}
  \hfill
  \subfloat[]{\includegraphics[height=7cm,width=8cm]{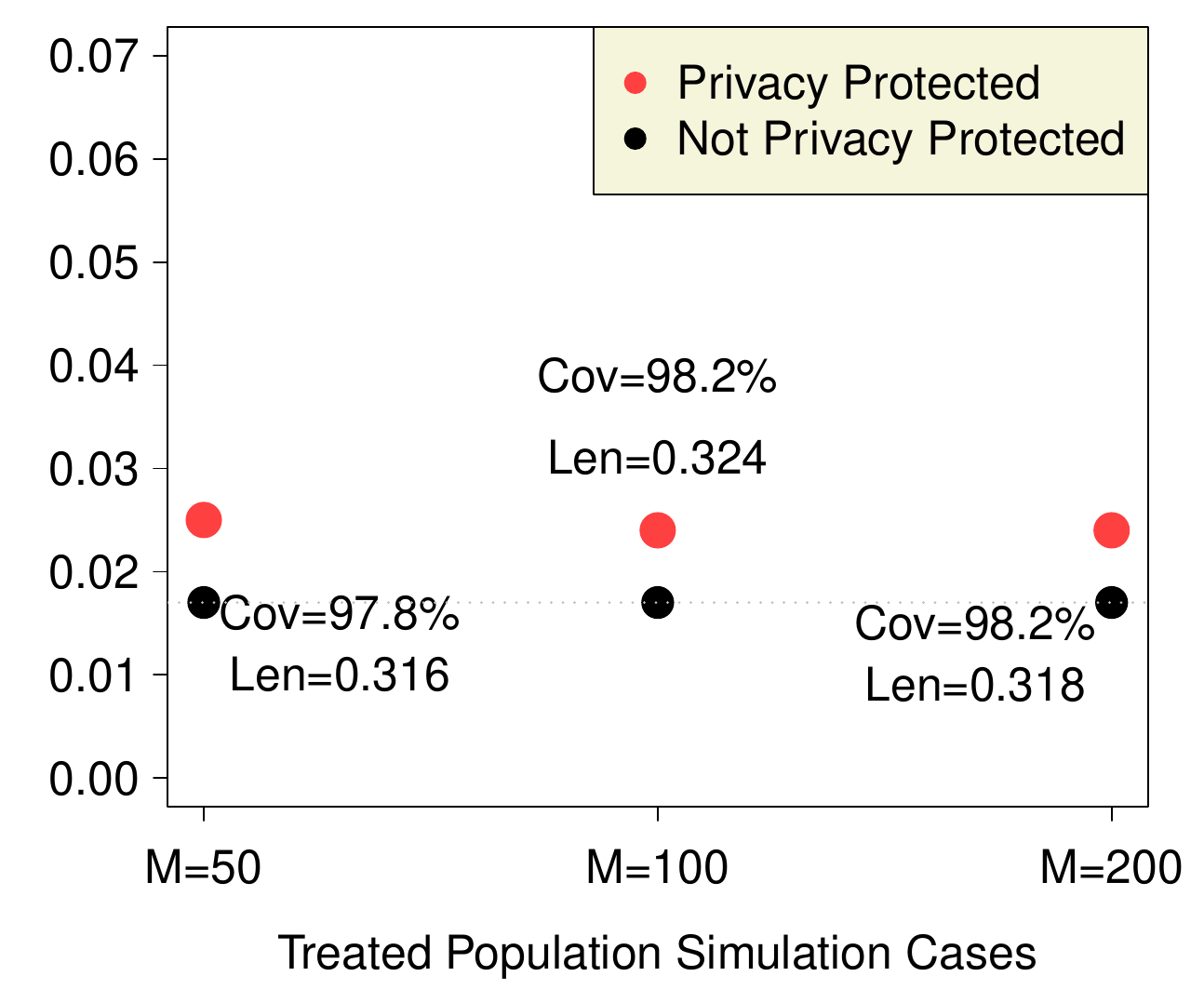}\label{sim-M2}}\\
  \subfloat[]{\includegraphics[height=7cm,width=8cm]{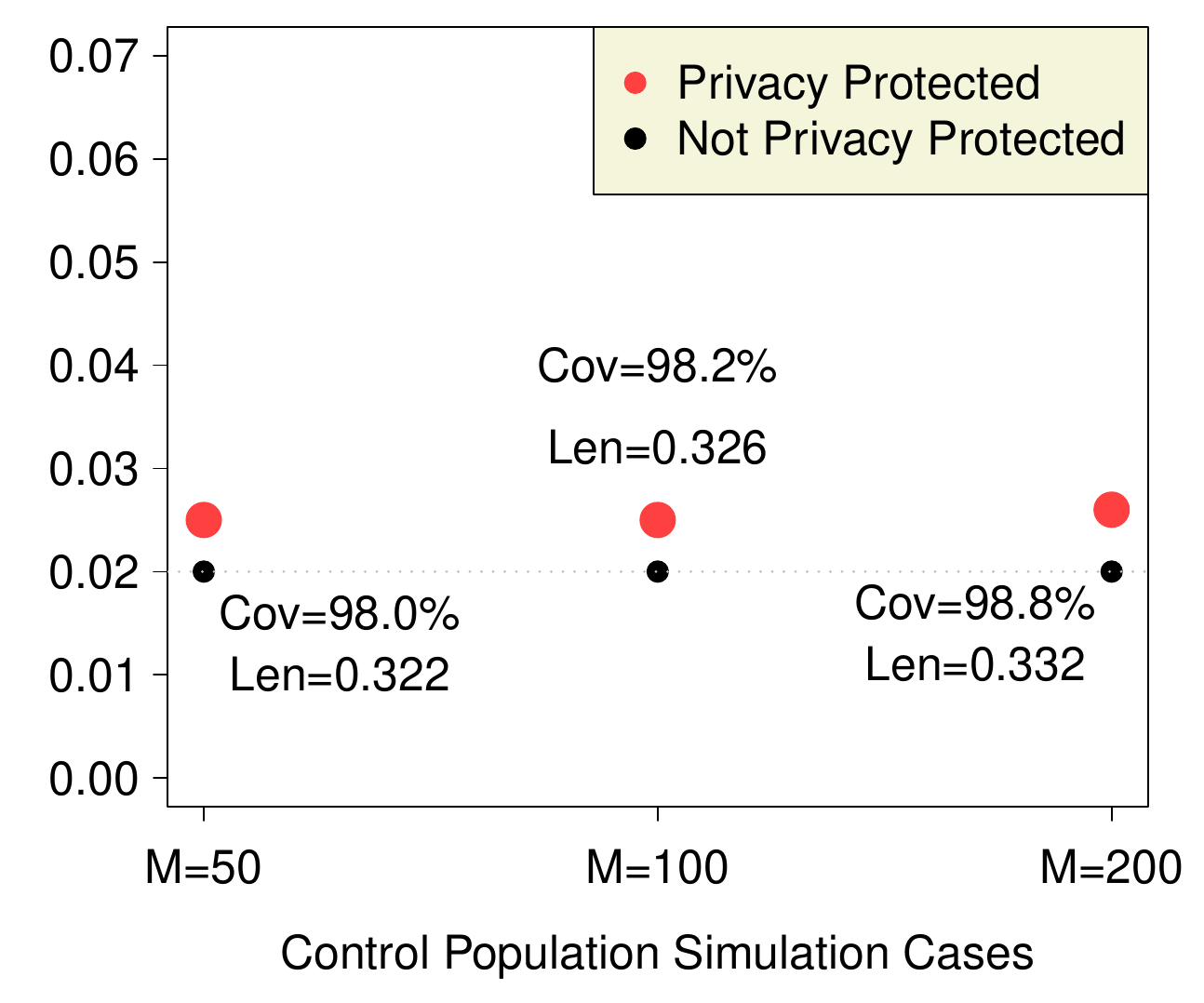}\label{sim-M3}}
  %\subfloat[]{\includegraphics[height=8cm,width=8cm]{Figures/M-Overlap.pdf}\label{sim-M4}}\\
  \caption{Root mean squared error (RMSE) of the differentially private WATEs for the ATE (Figure~\ref{sim-M1}), ATT (Figure~\ref{sim-M2}), and ATC  (Figure~\ref{sim-M3}), for $M \in \{50, 100, 200\}$. In all cases, $n=10000$, $\epsilon = 1$ and $a=.05$. RMSEs for WATE estimates with no privacy protection are denoted by black dots. RMSEs for the differentially private WATE estimates are denoted by red dots. ``Cov'' and ``Len'' stand for the coverage rate and average length of the 95\% intervals, both based on the differentially private algorithms.}
  \label{M-change}
\end{figure}

\subsection{Sensitivity to the Choice of Truncation Point}\label{sensitivity_a} 
While $a=.05$ may be a reasonable cut-off for propensity scores in many contexts, it is instructive to investigate the performance of the differentially private WATE inferences for other realistic values of $a$. To this end, we also examine the differentially private WATE inferences for $a\in\{0.03,0.07,0.1\}$.  Let $\tilde{\tau}^{\epsilon,a}$ correspond to the differentially private point estimate for the WATE for truncation limit $a$. We present the absolute distance between the non-private estimator on $\hat{\tau}$ computed using $\bD$, i.e., without truncation or privacy protections, and $\tilde{\tau}^{\epsilon,a}$. We write this difference as $\mbox{Dev}(\tilde{\tau}^{\epsilon,a})=|\tilde{\tau}^{\epsilon,a}-\hat{\tau}|$. We also present coverage rates and average lengths for the 95\% intervals 
using different choices of $a$. We set $\epsilon = 1$, $M = 100$, $\eta = 4$, and $\gamma = 1$. 

As evident in Table~\ref{Tab4}, in these simulations the properties of the point estimate and coverage rates for ATE, ATT and ATC are similar for all values of $a$ investigated.  This is because the Laplace noise variances are of comparable magnitude for the values of $a$  in this range. 
However, if we decrease $a$ to values near zero, say $a= .001$, the variance in the Laplace mechanism applied to $\bar{V}^T$ is about 20 times higher compared to using $a=0.05$. As a result, the 95\% intervals are much wider, and their coverage rate approaches $1$. When $a=0.001$, the RMSE values are nearly $1.5$ times higher than when $a=0.05$.

\begin{table}[!th]
%\scriptsize
\begin{center}
\begin{tabular}{ccccccc} 
\hline
\hline
& & & & & &  \\
& $\mbox{Dev}(\tilde{\tau}_{ATE}^{\epsilon,a})$ & $\mbox{Dev}(\tilde{\tau}_{ATT}^{\epsilon,a})$ & $\mbox{Dev}(\tilde{\tau}_{ATC}^{\epsilon,a})$ &  $\mbox{Cov}(\tilde{\tau}_{ATE}^{\epsilon,a})$ & $\mbox{Cov}(\tilde{\tau}_{ATT}^{\epsilon,a})$ & $\mbox{Cov}(\tilde{\tau}_{ATC}^{\epsilon,a})$ \\ 
& & & & & &  \\
\hline
\hline
& & & & & &  \\ 
 $a=0.03$ & 0.005 & 0.006 &  0.005 & 0.984 & 0.986 &  0.988 \\ 
  &  &  & & & & \\ 
%\cline{2-5} 
$a=0.07$ & 0.005 & 0.006 &  0.006 & 0.978 & 0.978 &  0.982\\ 
 &  &  & & & & \\ 
$a=0.10$ & 0.007 & 0.008 & 0.008 & 0.974 & 0.976 &  0.978 \\ 
 &  &  & & & & \\ 
\hline
\end{tabular}
\caption{Results for simulations for $a \in \{0.03, 0.07, 0.10\}$.  Entries include the average absolute difference in the differentially private WATE and the non-private WATE without truncation, labeled $\mbox{Dev}(\tilde{\tau}^{\epsilon,a})$, and the coverage rate of the 95\% intervals, labeled $\mbox{Cov}(\tilde{\tau}^{\epsilon,a}).$ In all cases, $M = 100$, $n=10000$, and $\epsilon=1$. }\label{Tab4}
\end{center}
\end{table}

\subsection{Sensitivity to the Sample Size}\label{sensitivity_n}
We next illustrate the effect of the  sample size by repeating simulations with  $n=100000$ and $n=5000$. 
We set $\epsilon=1$, $M=100$, $a=0.05$, $\eta = 4$, and $\gamma = 1$ for this simulation. 
%Similar to Table~\ref{Tab1}, T
Table~\ref{Tab_big_data} displays the results. The point estimates from the differentially private WATE estimators remain  accurate, with RMSE values decreasing as $n$ increases. 
When $n = 5000$, the interval estimates widen substantially for both the non-private and differentially private estimators. At $n=5000$, the coverage rate for the differentially private interval is near 100\%, indicating that the inferential procedure at this sample size results in overly wide intervals.  

\begin{table}[!th]
\small
\begin{center}
\begin{tabular}{lcccccc} 
\hline
& \multicolumn{3}{c}{$n=100000$} & \multicolumn{3}{c}{$n=5000$}\\
\hline
& Combined & Treated & Control & Combined & Treated & Control \\
\hline
\hline
 $\hat{\tau}$ \\
$\,\,\,\,$ RMSE & 0.012 & 0.013 &  0.013 & 0.022 & 0.023 & 0.025\\ 
$\,\,\,\,$ 95\% CI coverage & 0.928 & 0.918 & 0.920 & 0.970 & 0.972 & 0.972\\ 
$\,\,\,\,$ 95\% CI length & 0.165 & 0.189 & 0.184 & 0.291  & 0.345
&  0.343\\ 
\hline
$\tilde{\tau}^{\epsilon}$ \\
$\,\,\,\,$ RMSE & 0.017 & 0.018 & 0.018 & 0.028  &  0.029 & 0.031
 \\ 
$\,\,\,\,$ 95\% CI coverage & 0.972 & 0.976 & 0.978 &  0.996  & 0.998
  &  0.998\\ 
$\,\,\,\,$ 95\% CI length & 0.254 & 0.271 & 0.276 &  0.648  & 0.729
 & 0.742 \\
\hline
\hline
\end{tabular}
\caption{Results with sample sizes $n=100000$ and $n=5000$, with $(M=100, \epsilon = 1, a=.05)$. Here,  $\tilde{\tau}^{\epsilon}$ is the privacy protected treatment effect estimate. Results include coverage rates and average lengths of the 95\% intervals.} \label{Tab_big_data}
\end{center}
\end{table}

\subsection{Sensitivity to the value of $\epsilon$}\label{sensitivity_ep}

Finally, we consider the effects of changing $\epsilon$.  In general, the choice of $\epsilon$ is driven by privacy desiderata, for example, as specified by the data holders.  Ideally, the value of $\epsilon$ also takes into account the likely usefulness of the outputs that could  result from applying the differentially private algorithm; that is, the process of setting $\epsilon$ considers a trade off between risk and usefulness. 
Here, we examine results for $\epsilon \in \{0.5, 1, 5\}$. As in Section~\ref{sensitivity_M}, to generate $\bD$ in each of the $500$ simulation runs, we set $\eta = 4$ and $\gamma = 1$. 

Figure \ref{ep-change} displays the results. 
Changing $\epsilon$ in these simulations influences the estimated variances and hence lengths of the interval estimates. 
For small values of $\epsilon$, the differentially private  95\% intervals are  wider and have larger coverage rates. As $\epsilon$ increases, the coverage rates become closer to nominal. 

\begin{figure}[t]
  \centering
   \subfloat[]{\includegraphics[height=7cm,width=8cm]{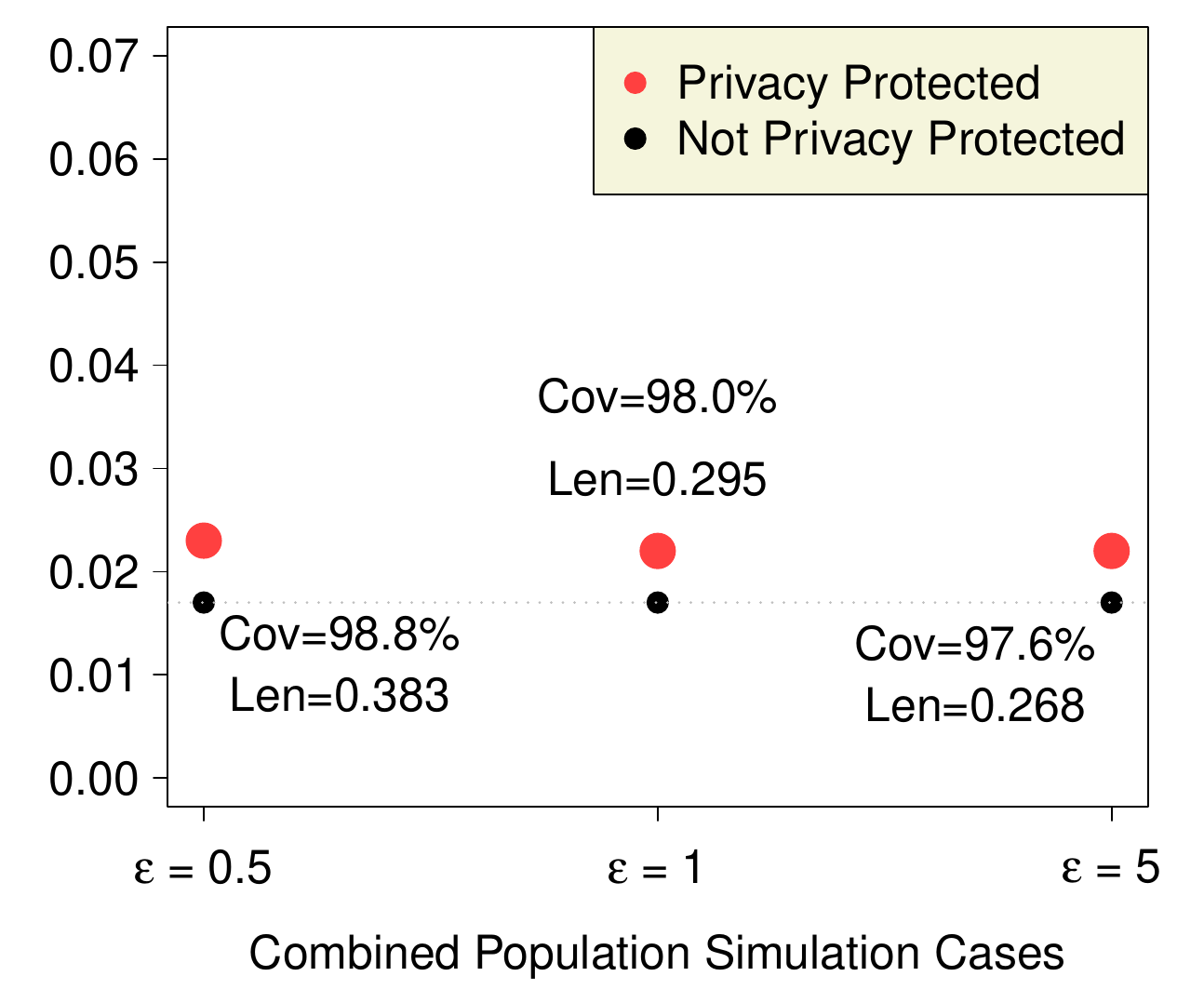}\label{sim-epsilon1}}
   \subfloat[]{\includegraphics[height=7cm,width=8cm]{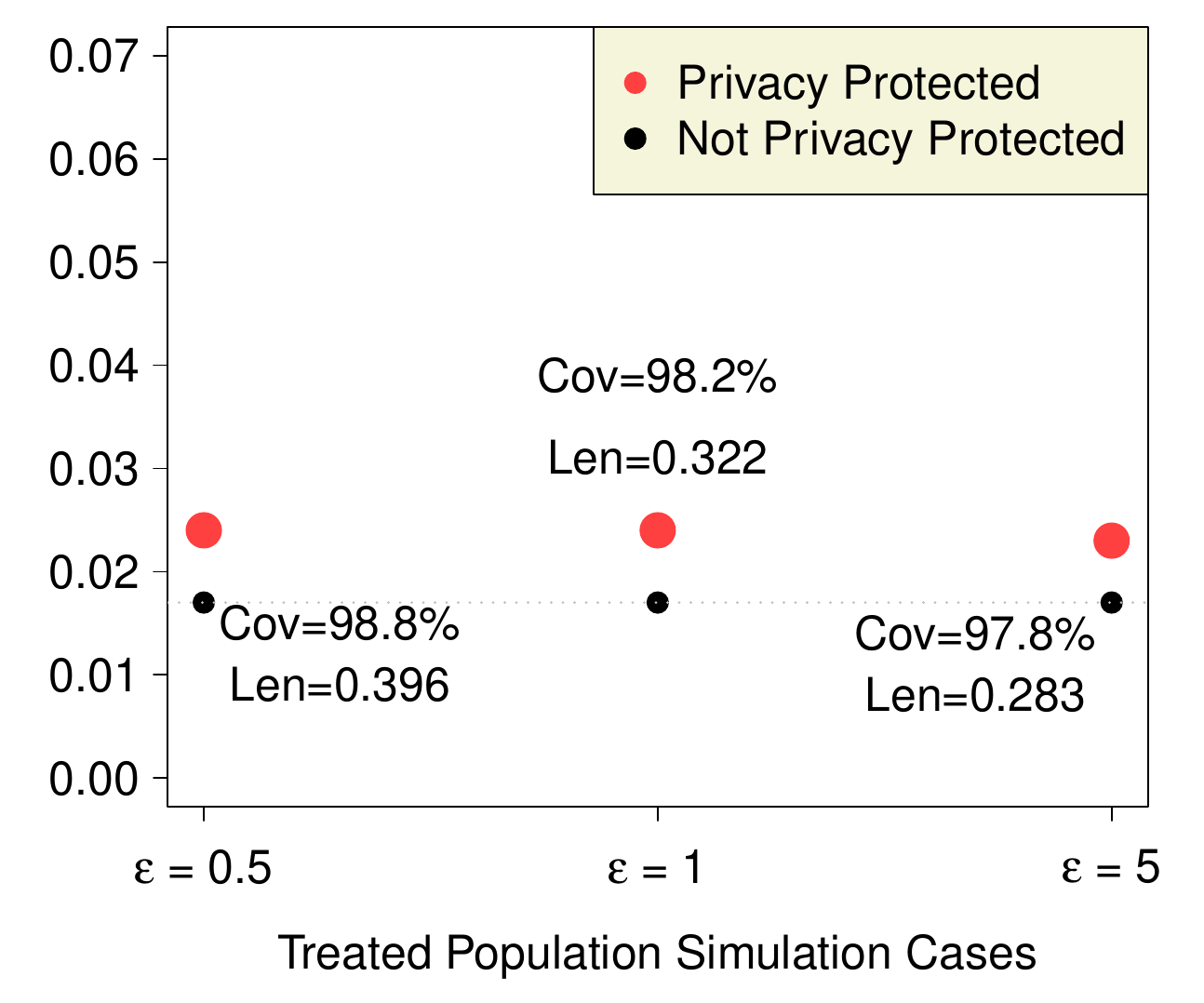}\label{sim-epsilon2}}\\
   \subfloat[]{\includegraphics[height=7cm,width=8cm]{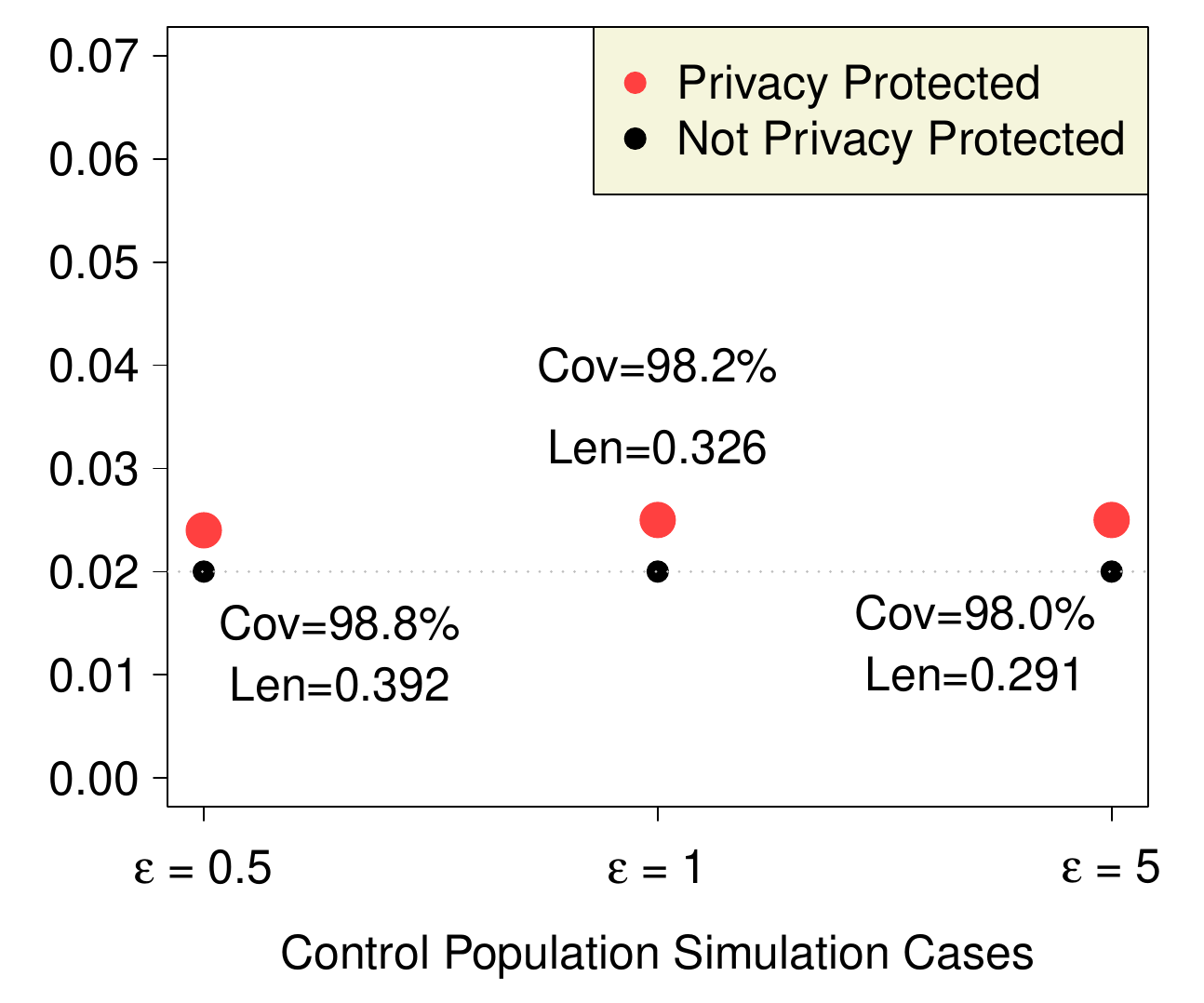}\label{sim-epsilon3}}
   %\subfigure[]{\includegraphics[height=8cm,width=8cm]{Figures/epsilon-Overlap.pdf}\label{sim-epsilon4}}\\
   \caption{Root mean squared error (RMSE) of the differentially private  WATEs for the ATE (Figure~\ref{sim-epsilon1}), ATT  (Figure~\ref{sim-epsilon2}), and ATC (Figure~\ref{sim-epsilon3}) for $\epsilon \in \{0.5, 1, 5\}.$  In all these cases,  $M = 100$, $n=10000$, and $a=0.05$.   RMSEs for the WATE estimates with no privacy protection are denoted by black dots.  RMSEs for the differentially private WATE estimates  are denoted by red dots.  ``Cov'' and ``Len'' stand for the coverage rate and average length of the 95\%  intervals, both based on the differentially private algorithms.}\label{ep-change}
\end{figure}

\section{Illustration with the Adult Income Data}\label{sec:real_data}

We demonstrate the application of the differentially private WATE estimators  using the ``Adult'' data set 
%which is an extract from the 1994 Census Bureau database 
%and contains information on adult earning
\citep{misc_adult_2}, which we accessed via the UCI Machine Learning Repository (\url{https://archive.ics.uci.edu/}). We emphasize that this example serves to illustrate the methodology and is not intended to be a thorough causal analysis of the effect of education on income.

The data comprise $n=30162$ individuals with complete information on the following variables: the age of the individual; the marital status, which includes seven categories - married to a civilian spouse, married to a spouse in the armed forces, married but with an absentee spouse, never married, divorced, separated, or widowed; the race, which includes five categories - white, black, American Indian Eskimo, Asian Pacific islander, and other; the sex, which includes two categories - male and female; the individual's occupation, which spans across 15 categories including executive-managerial, farming, fishing, transportation, sales, administrative-clerical, and more; and, an indicator of  whether their native country is the USA or not. We classify individuals as treated $(z = 1)$ if they have earned a bachelor's  degree or higher, and as controls $(z = 0)$ if they have an education level lower than a bachelor's degree. We make a binary outcome from income as below \$50000  ($y=0$) or at least \$50000 ($y=1$).

\begin{table}[t]
%\small
\begin{center}
\begin{tabular}
[c]{cccccc}
\hline
& $\epsilon$ & %$\tilde{\tau}_{ATE}^{\epsilon}$  
ATE & 
%$\tilde{\tau}_{ATT}^{\epsilon}$  
ATT &  ATC \\
%$\tilde{\tau}_{ATC}^{\epsilon}$  \\
&    &  95\% CI &  95\% CI &  95\% CI\\
\hline
% 100 & 5  &  0.256 & 0.271 & 0.261  \\
%  &  &  (0.196, 0.318) & (0.214, 0.324) & (0.204, 0.317) \\
%\hline
 $\tilde{\tau}^{\epsilon}$  & 1  & 0.271  & 0.258 & 0.236\\
&    &  (0.183, 0.361) & (0.170, 0.344) & (0.149, 0.324)\\
%\hline
&  0.5 & 0.272  &  0.269 & 0.275 \\
 &  &  (0.144, 0.401) &  (0.136, 0.406) & (0.151, 0.407)\\
$\hat{\tau}$  & &   \multicolumn{1}{c}{0.263} &  \multicolumn{1}{c}{0.271} &  \multicolumn{1}{c}{0.260}  \\
   %\multicolumn{3}{c}{95\% CI} 
   & &   \multicolumn{1}{c}{(0.251, 0.275)} &  \multicolumn{1}{c}{(0.259, 0.283)} &  \multicolumn{1}{c}{(0.248, 0.272)} \\
\hline
\end{tabular}
\caption{Point estimates and 95\% intervals for treatment effects from the analysis of the Adult Income Data described in Section \ref{sec:real_data}. Results include the differentially private inferences based on Algorithm \ref{alg1} in the panel indicated by $\tilde{\tau}^{\epsilon}$, as well as results based on the full data with no privacy protection or truncation in the panel headed by $\hat{\tau}$.  Differentially private results use $M=100$ and $a= 0.05$.}
\label{Tab_real_Education}
\end{center}
\end{table}

We estimate propensity scores using a logistic regression of $z$ on {age, marital status, race, sex, occupation, and a binary indicator denoting whether the individual's native country is the USA}. We use Algorithm \ref{alg1} to make differentially private inferences about the effect of education on income.   
Given the relatively large sample size, for this analysis, we set $M=100$.  This $M$ value is sufficiently large enough that the Laplace noise variance is expected to be less than the within-partition variance of the WATE.  We set $a= 0.05$ and $\pi=0.5$.

Table~\ref{Tab_real_Education} presents results for two values of $\epsilon\in\{0.5,1\}$. 
For both the differentially private WATE estimates and their non-privacy protected counterparts, the point estimates are positive and 95\% confidence intervals exclude zero. These suggest a positive association between years of education and income level. The differentially private WATE estimates and their non-privacy protected analogues are close to one another. The 95\% intervals constructed from Algorithm \ref{alg1} are wider than the corresponding intervals based on the full data.

\section{Concluding Remarks}\label{sec:conc}
We present an  approach to estimating weighted average treatment effects with binary outcomes while ensuring differential privacy.  Simulation and empirical results suggest that the approach can result in accurate point estimates with conservative interval coverage rates. To implement the approach, analysts need to select $M$ and $a$. We now suggest some rule-of-thumb guidance for these choices.  We note that analysts can undertake simulation studies akin to those presented here to facilitate choices tuned more closely to their settings.

In determining appropriate values for $M$, we recommend analysts consider three desiderata.  First, they 
should aim to achieve an acceptable margin of error $\delta$  (i.e., two times standard error) for the treatment effect estimate of interest.
%if they were to use $\bar{\tau}^T$ as its estimator.  
  To illustrate, consider ${\tau}_{ATE}$. Absent privacy concerns, the approximate standard error of $\bar{\tau}^T$ for the ATE is  expected not to exceed $\sqrt{1/(2an)}.$ Similarly, the standard deviation of independently drawn noise from the Laplace distribution with sensitivity $2/M$  and privacy budget $\epsilon(1-\pi)$ is $2/(M\epsilon(1-\pi))$.  We can solve for the positive $M$ that makes $\delta \approx 2\sqrt{1/(2an) + 4/(M\epsilon(1-\pi))^2}$, resulting in the formula
  \begin{equation}
  M \approx \left(\frac{2}{\epsilon(1-\pi)}\right) \left(\frac{1}{\sqrt{\delta^2/4 - 1/(2an)}}\right).\label{Meq}
  \end{equation}
In cases where $n$ is large enough that $1/(2an)$ does not contribute meaningfully to the computation in \eqref{Meq}, we can simplify this as $M \approx 4/(\epsilon(1-\pi)\delta)$.  For example, in the ADULTS data analysis with $n=30162$, $a=.05$ (USE ,10??), and $\epsilon(1-\pi) = 0.5$, for a desired $MoE = .10$, we solve for $M\approx80$.  Second, analysts should check that $M$ is small enough that it is unlikely to generate a set of random partitions where one or more has too few treated or control cases to estimate the propensity scores; see below.  Finally, the analyst should make $M$ large enough that results are expected to be reasonably stable over different runs of the algorithm on the data (although only one run with $\bD$ is expected to minimize privacy loss). We propose using a minimum of $M=50$ when possible.  
In general, we recommend selecting $a$ to be sufficiently small to minimize the likelihood of truncating propensity scores, yet not so small as to inflate the sensitivity of the variance estimators.  In essence, $a$ can be adjusted to ensure that the sensitivity of the variance estimator is  acceptable, particularly for the values of $M$ under consideration. As a default, we suggest setting $a=0.1$. This value is recommended by \cite{crump2009dealing} and empirically supported by \cite{sturmer} as a default value for trimming propensity scores outside the privacy context. 

As noted in Section \ref{sec:DPwate}, it is theoretically possible for the partitioning to result in failures, for example, in one or more partitions all records have $z_i=1$. When this occurs one may not be able to estimate the propensity scores or the treatment effects. If the expectation is that such partition assignments are rare---as can be expected in large samples with ample treated and control cases---we suggest re-doing the partitioning with the expectation of sampling a more suitable set of partitions.  If the expectation is that this is not rare, one possibility is to set the undefined $\hat{\tau}_m^{T}$ equal to random draws from the uniform distribution over $[-1, 1]$ and the undefined $\hat{V}_m^{T}$ equal to random draws from the uniform distribution over 0 to the appropriate  sensitivity from Theorem \ref{th:variance}.  We conjecture that this should preserve the privacy guarantee without obliterating the usefulness of the inferences.  Investigating this conjecture is an area for future work.

\section{Appendix}
This section presents proofs of the theoretical results outlined in Section~\ref{theory}.

\subsection{Proof of Lemma \ref{lemma_DP}}

%\begin{proof} 
We begin by proving part (i) of the lemma. Note that, 
\begin{align}\label{theory_eq1}
&P(|\bar{\tau}-\tau|>c)
\leq P(|\bar{\tau}-\bar{\tau}^{T,\epsilon}|>c/3)\nonumber\\
&\qquad\qquad\qquad\qquad+P(|\bar{\tau}^{T}-\bar{\tau}^{T,\epsilon}|>c/3)
+P(|\bar{\tau}^{T}-\tau|>c/3).
\end{align}
The first and second terms 
correspond to probabilities under the Laplace distribution.  More specifically,
\begin{align}\label{theory_eq2}
&P(|\bar{\tau}^{T}-\bar{\tau}^{T,\epsilon}|>c/3)=E_{y,\bx,z}P(|\bar{\tau}^{T}-\bar{\tau}^{T,\epsilon}|>c/3|\bD)
=\exp(-M\epsilon(1-\pi)c/6)\nonumber\\
&P(|\bar{\tau}-\bar{\tau}^{T,\epsilon}|>c/3)=E_{y,\bx,z}P(|\bar{\tau}-\bar{\tau}^{T,\epsilon}|>c/3|\bD)\leq\exp(-M\epsilon(1-\pi)c/6).
\end{align}
It remains to show the bound for $P(|\bar{\tau}^{T}-\tau|>c/3)$. To this end, note that
\begin{align}\label{differential_theory}
&(c^2/9) P(|\bar{\tau}^{T}-\tau|>c/3)
\leq E_{y,\bx,z}[(\bar{\tau}^{T}-\tau)^2]
\leq 2 Var(\bar{\tau}^{T})+2\{E_{y,\bx,z}[\bar{\tau}^{T}]-\tau\}^2\nonumber\\
&=(2/M^2)\sum_{m=1}^M Var(\hat{\tau}_{m}^T)+2\{(1/M)\sum_{m=1}^ME_{y,\bx,z}[\hat{\tau}_{m}^T]-\tau\}^2.
\end{align}
Let $\bD_{u,m}, \bD_{l,m}, \bD_{T,m}$ denote the samples within the $m$th partition $\bD_m$ with propensity score greater than $1-a$, less than $a$ and between $a$ and $(1-a)$, respectively.
Now observe that 
\begin{align}\label{eq:three_terms}
&\frac{1}{n_{m}}\sum_{i=1}^{n_m}w_{1i}z_iy_i(1)   
=\frac{1}{n_{m}}\sum_{i\in \bD_{T,m}}w_{1i}z_iy_i(1)+\frac{1}{n_{m}}\sum_{i\in\bD_{l,m}}w_{1i}z_iy_i(1)+\frac{1}{n_{m}}\sum_{i\in\bD_{u,m}}w_{1i}z_iy_i(1)\nonumber\\
&=\frac{1}{n_{m}}\sum_{i\in \bD_{T,m}}\frac{t(\bx_i)}{e^T(\bx_i)}z_iy_i(1)+\frac{1}{n_{m}}\sum_{i\in\bD_{l,m}}\frac{t(\bx_i)}{e^T(\bx_i)}z_iy_i(1)+\frac{1}{n_{m}}\sum_{i\in\bD_{u,m}}\frac{t(\bx_i)}{e^T(\bx_i)}z_iy_i(1)\nonumber\\
&=\frac{1}{n_{m}}\sum_{i\in \bD_{T,m}}\frac{t(\bx_i)}{e(\bx_i)}z_iy_i(1)+\frac{1}{n_{m}a}\sum_{i\in\bD_{l,m}}t(\bx_i)z_iy_i(1)+\frac{1}{n_{m}(1-a)}\sum_{i\in\bD_{u,m}}t(\bx_i)z_iy_i(1).
\end{align}

By assumptions (A1) and (A2),
\begin{align}\label{WATE_proof_eq3}
&\frac{1}{n_{m}}\sum_{i\in\bD_{T,m}}w_{1i}z_iy_i(1)\stackrel{a.s.}{\rightarrow}
E_{y,\bx,z}[y(1)zt(\bx)/e(\bx)]=E_{\bx}E_{z|\bx}E_{y|z,\bx}[y(1)zt(\bx)/e(\bx)]\nonumber\\
&=E_x[E[y(1)|\bx]t(\bx)]=\int E[(y(1)|\bx]t(\bx)f(\bx)\Delta(d\bx).
\end{align}
Since $|t(\bx_i)z_iy_i(1)|\leq 1$, for all $i$,
$\frac{1}{n_{m}a}\sum_{i\in\bD_{l,m}}t(\bx_i)z_iy_i(1)\stackrel{a.s.}{\rightarrow} 0$ if $n_{l,m}$ is finite, and\\ $\frac{1}{n_{m}(1-a)}\sum_{i\in\bD_{u,m}}t(\bx_i)z_iy_i(1)\stackrel{a.s.}{\rightarrow} 0$ if $n_{u,m}$ is finite. When both $n_{l,m}$ and $n_{u,m}$ are infinite, 
\begin{align}\label{eq:second_term}
\frac{1}{n_{m}a}\sum_{i\in\bD_{l,m}}t(\bx_i)z_iy_i(1)=\frac{n_{l,m}}{n_{m}}\frac{1}{n_{l,m}a}\sum_{i\in\bD_{l,m}}t(\bx_i)z_iy_i(1)\stackrel{a.s.}{\rightarrow} 0, 
\end{align}
as $\frac{1}{n_{l,m}}\sum_{i\in\bD_{l,m}}t(\bx_i)z_iy_i(1)\stackrel{a.s.}{\rightarrow} \int E[(y(1)|\bx]e(\bx)t(\bx)f(\bx)\Delta(d\bx)$ and $\frac{n_{l,m}}{n_{m}}\rightarrow 0$ by (A2). Using the similar logic, 
\begin{align}\label{eq:third_term}
\frac{1}{n_{m}}\sum_{i\in\bD_{u,m}}t(\bx_i)z_iy_i(1)\stackrel{a.s.}{\rightarrow} 0. 
\end{align}
From (\ref{eq:three_terms}), (\ref{WATE_proof_eq3}), (\ref{eq:second_term}) and (\ref{eq:third_term}), we have
\begin{align}\label{eq:first_term_full} 
\frac{1}{n_{m}}\sum_{i=1}^{n_m}t(\bx_i)z_iy_i(1)\rightarrow \int E[(y(1)|\bx]t(\bx)f(\bx)\Delta(d\bx).
\end{align}

Also,
\begin{align}
\frac{1}{n_{m}}\sum_{i=1}^{n_{m}}w_{0i}(1-z_i)y_i(0)
&=\frac{1}{n_{m}}\sum_{i\in \bD_{T,m}}w_{0i}(1-z_i)y_i(0)+\frac{1}{n_{m}}\sum_{i\in\bD_{l,m}}w_{0i}(1-z_i)y_i(0)\nonumber\\
&+\frac{1}{n_{m}}\sum_{i\in\bD_{u,m}}w_{0i}(1-z_i)y_i(0).    
\end{align}
Following the similar arguments as above,
\begin{align}\label{WATE_proof_eq1}
&\frac{1}{n_{m}}\sum_{i\in \bD_{T,m}}w_{0i}(1-z_i)y_i(0)\stackrel{a.s.}{\rightarrow}
E_{y,\bx,z}[y(0)(1-z)t(\bx)/(1-e(\bx))]
=\int E[(y(0)|\bx]t(\bx)f(\bx)\Delta(d\bx).\nonumber\\
&\frac{1}{n_{m}}\sum_{i\in\bD_{l,m}}w_{0i}(1-z_i)y_i(0)\stackrel{a.s.}{\rightarrow} 0,\:\:\:
\frac{1}{n_{m}}\sum_{i\in\bD_{u,m}}w_{0i}(1-z_i)y_i(0)\stackrel{a.s.}{\rightarrow} 0
\end{align}
We use the same argument as above to arrive at
\begin{align}\label{WATE_proof_eq2}
&\frac{1}{n_{m}}\sum_{i=1}^{n_{m}}w_{1i}z_i\stackrel{a.s.}{\rightarrow}
E_{y,\bx,z}[zt(\bx)/e(\bx)]=E_x[t(\bx)]=\int t(\bx)f(\bx)\Delta(d\bx)\nonumber\\
&\frac{1}{n_{m}}\sum_{i=1}^{n_{m}}w_{0i}(1-z_i)\stackrel{a.s.}{\rightarrow}
E_{y,\bx,z}[(1-z)t(\bx)/(1-e(\bx))]=E_x[t(\bx)]=\int t(\bx)f(\bx)\Delta(d\bx).
\end{align}
Hence,
$\hat{\tau}_{m}^T\stackrel{a.s.}{\rightarrow}\left[\int \{E[y(1)|\bx]-E[y(0)|\bx]\} t(\bx)f(\bx)\Delta(d\bx)\right]/\left[\int t(\bx)f(\bx)\Delta(d\bx)\right]=\tau$ from (\ref{eq:first_term_full}), (\ref{WATE_proof_eq1}) and (\ref{WATE_proof_eq2}). Given that each $\hat{\tau}_m$ is bounded between $-1$ to $1$, dominated convergence theorem leads to $E[\hat{\tau}_{m}^T]\rightarrow\tau$, as $n\rightarrow\infty$.
Following the proof of Theorem 2 in \cite{li2018balancing}, we have
\begin{align*}
n_m Var_{y|\bx,z}(\hat{\tau}_m^T)&=\frac{\frac{1}{n_m}\sum_{i=1}^{n_m}v_1(\bx_i)z_iw_{1i}^2}{[\frac{1}{n_m}\sum_{i=1}^{n_m}z_iw_{1i}]^2}+\frac{\frac{1}{n_m}\sum_{i=1}^{n_m}v_0(\bx_i)(1-z_i)w_{0i}^2}{[\frac{1}{n_m}\sum_{i=1}^{n_m}(1-z_i)w_{0i}]^2}\nonumber\\
 &=\frac{\frac{1}{n_m}\sum_{i=1}^{n_m}v_1(\bx_i)z_it(\bx_i)^2/e^T(\bx_i)^2}{[\frac{1}{n_m}\sum_{i=1}^{n_m}z_it(\bx_i)/e^T(\bx_i)]^2}+\frac{\frac{1}{n_m}\sum_{i=1}^{n_m}v_0(\bx_i)(1-z_i)t(\bx_i)^2/(1-e^T(\bx_i))^2}{[\frac{1}{n_m}\sum_{i=1}^{n_m}(1-z_i)t(\bx_i)/(1-e^T(\bx_i))]^2}.
\end{align*}
Note that 
\begin{align*}
\frac{1}{n_m}\sum_{i=1}^{n_m}v_1(\bx_i)z_i\frac{t(\bx_i)^2}{e^T(\bx_i)^2}
&=\frac{1}{n_m}\sum_{i\in\bD_{T,m}}v_1(\bx_i)z_i\frac{t(\bx_i)^2}{e(\bx_i)^2}
+\frac{1}{n_ma^2}\sum_{i\in\bD_{l,m}}v_1(\bx_i)z_it(\bx_i)^2\\
&+\frac{1}{n_m(1-a)^2}\sum_{i\in\bD_{u,m}}v_1(\bx_i)z_it(\bx_i)^2.
\end{align*}
By assumptions (A1) and (A2), and using the arguments used before,
$\frac{1}{n_m}\sum_{i\in\bD_{l,m}}v_1(\bx_i)z_it(\bx_i)^2\stackrel{a.s.}{\rightarrow} 0$ and $\frac{1}{n_m}\sum_{i\in\bD_{u,m}}v_1(\bx_i)z_it(\bx_i)^2\stackrel{a.s.}{\rightarrow} 0$, and 
$E_{\bx,z}[\frac{1}{n_m}\sum_{i\in\bD_{T,m}}v_1(\bx_i)z_i\frac{t(\bx_i)^2}{e(\bx_i)^2}]\rightarrow \int t(\bx)^2\frac{v_1(\bx)}{e(\bx)}f(\bx)\Delta(\bx)$.

\noindent Similarly,
\begin{align*}
\frac{1}{n_m}\sum_{i=1}^{n_m}v_0(\bx_i)\frac{(1-z_i)t(\bx_i)^2}{(1-e^T(\bx_i))^2}
&=\frac{1}{n_m}\sum_{i\in\bD_{T,m}}v_1(\bx_i)\frac{(1-z_i)t(\bx_i)^2}{(1-e^T(\bx_i))^2}
+\frac{1}{n_m(1-a)^2}\sum_{i\in\bD_{l,m}}v_1(\bx_i)(1-z_i)t(\bx_i)^2\\
&+\frac{1}{n_ma^2}\sum_{i\in\bD_{u,m}}v_1(\bx_i)(1-z_i)t(\bx_i)^2.
\end{align*}
Hence,
$\frac{1}{n_m}\sum_{i\in\bD_{l,m}}v_0(\bx_i)(1-z_i)t(\bx_i)^2\stackrel{a.s.}{\rightarrow} 0$ and $\frac{1}{n_m}\sum_{i\in\bD_{u,m}}v_0(\bx_i)(1-z_i)t(\bx_i)^2\stackrel{a.s.}{\rightarrow} 0$, and 
$E_{\bx,z}[\frac{1}{n_m}\sum_{i\in\bD_{T,m}}v_0(\bx_i)(1-z_i)\frac{t(\bx_i)^2}{(1-e(\bx_i))^2}]\rightarrow \int t(\bx)^2\frac{v_0(\bx)}{1-e(\bx)}f(\bx)\Delta(\bx)$. Using similar arguments,
\begin{align*}
 \frac{1}{n_m}\sum_{i=1}^{n_m}z_it(\bx_i)/e(\bx_i)\rightarrow \int t(\bx)f(\bx)\Delta(\bx),\:\:\frac{1}{n_m}\sum_{i=1}^{n_m}(1-z_i)t(\bx_i)/(1-e(\bx_i))\rightarrow \int t(\bx)f(\bx)\Delta(\bx).  
\end{align*}
Using Slutsky's theorem,
\begin{align*}
n_m E_{\bx}[Var(\hat{\tau}_{m}^T|\bx)]
\rightarrow\int t(\bx)^2\left\{\frac{v_1(\bx)}{e(\bx)}+\frac{v_0(\bx)}{(1-e(\bx))}\right\}f(\bx)\Delta(\bx)/\left\{\int t(\bx)f(\bx)\Delta(\bx)\right\}^2.
\end{align*}
As $n_{m}\rightarrow\infty$ $E_{\bx}[Var(\hat{\tau}_{m}^T|\bx)]\rightarrow 0$. Following \cite{imbens2004nonparametric}, typically, $Var_{\bx}(E[\hat{\tau}_{m}^T|\bx)])\leq E_{\bx}[Var(\hat{\tau}_{m}^T|\bx)]\rightarrow 0$. Hence, 
$Var_{y,\bx,z}(\hat{\tau}_{m}^T)\rightarrow 0$. 
%\end{proof}

Combining equations (\ref{theory_eq1}), (\ref{theory_eq2}), (\ref{differential_theory}) and the result above, as $n\rightarrow\infty$
\begin{align}
P(|\bar{\tau}-\tau|>c)\leq 2\exp(-M\epsilon(1-\pi)c/6).
\end{align}
which concludes the proof of (i).

To prove (ii), note that
\begin{align*}
& P(|\bar{V}-V|>c)\leq P(|\bar{V}-\bar{V}^{T,\epsilon}|>c/3)\nonumber\\
&\qquad\qquad\qquad\qquad+P(|\bar{V}^{T}-\bar{V}^{T,\epsilon}|>c/3)
+P(|\bar{V}^{T}-V|>c/3).   
\end{align*}

The first and second terms are straightforward to bound following the Laplace distribution. More specifically,
\begin{align}\label{theory_eq_var}
&P(|\bar{V}^{T}-\bar{V}^{T,\epsilon}|>c/3)=E_{y,\bx,z}P(|\bar{V}^{T}-\bar{V}^{T,\epsilon}|>c/3|\bD)
=\exp(-Mc\epsilon\pi/6)\nonumber\\
&P(|\bar{V}-\bar{V}^{T,\epsilon}|>c/3)=E_{y,\bx,z}P(|\bar{V}-\bar{V}^{T,\epsilon}|>c/3|\bD)\leq\exp(-Mc\epsilon\pi/6).
\end{align}
Regarding the third term, note that $\bar{V}^T=\sum_{m=1}^M\hat{V}_{m}^T/M$ and 
$\hat{V}_{m}^T\stackrel{P}{\rightarrow} V$ as $n_{m}\rightarrow \infty$, using the above results. Hence, $\bar{V}^T\stackrel{P}{\rightarrow} V$ which implies
$P(|\bar{V}^{T}-V|>c/3)\rightarrow 0$ as $n\rightarrow\infty$. Hence, 
$P(|\bar{V}^{T}-V|>c/3)\leq 2 \exp(-M\epsilon\pi c/6)$ as $n\rightarrow\infty$, proving (ii).
%\end{proof}

\subsection{Proof of Theorem \ref{theorem_DP}}

We have 
\begin{align*}
 KL(\tilde{g}^{\epsilon},g)=\frac{(\bar{\tau}-\tau)^2}{2V}+\frac{\bar{V}}{2V}-\frac{1}{2}-\frac{1}{2}\log\frac{\bar{V}}{V}=U_1+U_2+U_3,
\end{align*}
where $U_1=\frac{(\bar{\tau}-\tau)^2}{2V}$, $U_2=\frac{\bar{V}}{2V}-\frac{1}{2}$ and $U_3=-\frac{1}{2}\log\frac{\bar{V}}{V}$.
Following Lemma~\ref{lemma_DP}, for any $c>0$, as $n\rightarrow\infty$,
\begin{align}
& P(|U_1|>c/3)=P\left(\frac{(\bar{\tau}-\tau)^2}{2V}>c/3\right)
\leq 2\exp\left(-\frac{M\epsilon(1-\pi)\sqrt{2Vc}}{6\sqrt{3}}\right)\\
& P(|U_2|>c/3)=P\left(\Big|\frac{\bar{V}}{2V}-\frac{1}{2}\Big|>c/3\right)
\leq 2\exp(-M\epsilon\pi V c/9)\\
& P(|U_3|>c/3)=P\left(\Big|\frac{1}{2}\log\frac{\bar{V}}{V}\Big|>c/3\right)\nonumber\\
&\qquad\qquad\qquad\qquad\leq P\left(\Big|\frac{\bar{V}}{2V}-\frac{1}{2}\Big|>c/3\right)
\leq 2\exp(-M\epsilon\pi V c/9),
\end{align}
where last inequality uses the fact that $log(h)\leq h-1$, for any $h>0$.
Finally, for any $c>0$,
\begin{align*}
P(KL(\tilde{g}^{\epsilon},g)>c)&\leq 
P(|U_1|>c/3)+P(|U_2|>c/3)+P(|U_3|>c/3)\\
&\leq 2\exp\left(-\frac{M\epsilon(1-\pi)\sqrt{2Vc}}{6\sqrt{3}}\right)+4\exp(-M\epsilon\pi V c/9),
\end{align*}
as $n\rightarrow\infty$, proving the result. %Hence,

\bibliographystyle{natbib}
\bibliography{references}
\end{document}